\documentclass[lettersize,journal]{IEEEtran}
\usepackage{amsmath,amssymb}

\usepackage[hidelinks]{hyperref}  
\usepackage{hyperref}
\usepackage{bm}
\usepackage{graphicx,graphics,color,psfrag}
\usepackage{cite,balance}
\usepackage{caption}
\allowdisplaybreaks
\usepackage{algorithm}
\usepackage{algorithmic}
\usepackage{accents}
\usepackage{amsthm}
\usepackage{url}
\usepackage[english]{babel}
\usepackage{multirow}
\usepackage{enumerate}
\usepackage{cases}
\usepackage{stfloats}
\usepackage{subfigure}
\usepackage{dsfont}
\usepackage{color,soul}
\usepackage{amsfonts}
\usepackage{cite,graphicx,amsmath,amssymb}
\usepackage{fancyhdr}
\usepackage{hhline}
\usepackage{graphicx,graphics}
\usepackage{array,color}
\usepackage{mathtools}
\usepackage{amsmath}

\newtheorem{definition}{\emph{\underline{Definition}}}

\newtheorem{lemma}{\emph{\underline{Lemma}}}

\newtheorem{remark}{\bf \emph{\underline{Remark}}}
\newtheorem{example}{\bf \emph{\underline{Example}}}

\def\({\left(}
\def\){\right)}

\def\b0{{\mathbf{0}}}

\renewcommand{\mod}{\tx{mod}}

\newcommand{\tx}[1]{\texttt{#1}}

\newcommand{\tr}{\mathrm{tr}}
\newcommand{\diag}{\mathrm{diag}}

\newcommand{\bet}{\boldsymbol{\eta}}
\newcommand{\bmu}{\boldsymbol{\mu}}

\begin{document}
	\captionsetup[figure]{name={Fig.},labelsep=period,singlelinecheck=off}  
	\title{Near-field Target Localization: Effect of \\ Hardware Impairments}
	\author{Jiapeng~Li,~\IEEEmembership{Student~Member,~IEEE}, 
		Changsheng~You,~\IEEEmembership{Member,~IEEE},  Chao~Zhou,~\IEEEmembership{Student~Member,~IEEE}, 
		Yong~Zeng,~\IEEEmembership{Fellow,~IEEE}, and Zhiyong~Feng,~\IEEEmembership{Senior~Member,~IEEE}\thanks{Part of this work has be presented in 2025 IEEE Global Communications Conference (GLOBECOM)~\cite{Globecom_2025}. 
			Jiapeng Li, Changsheng You and Chao Zhou are with the Department of Electronic and Electrical Engineering, Southern University of Science and Technology, Shenzhen 518055, China (e-mail: \{lijiapeng2023, zhouchao2024\}@mail.sustech.edu.cn; youcs@sustech.edu.cn).
			Yong Zeng is with the National Mobile Communications Research Laboratory, Southeast University, Nanjing 210096, China, also with the Purple Mountain Laboratories, Nanjing 211111, China (e-mail: yong$\_$zeng@seu.edu.cn).
			Zhiyong Feng is with the Key Laboratory of Universal Wireless Communications, Ministry of Education, Beijing University of Posts and Telecommunications, Beijing 100876, China (e-mail: fengzy@bupt.edu.cn).
			(\emph{Corresponding author: Changsheng You.})}\vspace{-20pt}}
	\maketitle

	\begin{abstract}
		The prior works on near-field target localization have mostly assumed ideal hardware models and thus suffer from two limitations in practice. 
		First, extremely large-scale arrays (XL-arrays) usually face a variety of \emph{hardware impairments} (HIs) that may introduce unknown phase and/or amplitude errors. 
		Second, the existing block coordinate descent (BCD)-based methods for joint estimation of the HI indicator, channel gain, angle, and range may induce considerable target localization error when the target is very close to the XL-array. 
		To address these issues, we propose in this paper a new three-phase \emph{HI-aware near-field localization} method, by efficiently detecting faulty antennas and estimating the positions of targets. 
		Specifically, we first determine faulty antennas by using compressed sensing (CS) methods and improve detection accuracy based on coarse target localization.
		Then, a dedicated phase calibration method is designed to correct phase errors induced by detected faulty antennas. Subsequently, an efficient near-field localization method is devised to accurately estimate the positions of targets based on the full XL-array with phase calibration.
		Additionally, we resort to the misspecified Cramér-Rao bound (MCRB) to quantify the performance loss caused by HIs.
		Last, numerical results demonstrate that our proposed method significantly reduces the localization errors as compared to various benchmark schemes, especially for the case with a short target range and/or a high fault probability. 
	\end{abstract}\vspace{-4pt}
	
	\vspace{-8pt}
	\section{Introduction} \vspace{-2pt}
	\emph{Extremely large-scale arrays} (XL-arrays) have emerged~as a promising technology to meet the demands for significantly high spectral efficiency and spatial resolution in the sixth-generation (6G) wireless networks~\cite{CuiNear2023,10496996,10663521,10540249,10220205}. 
	In~particular, as the number of antennas increases, the Rayleigh distance, specifying the boundary between near-field and far-field, is greatly expanded~\cite{10858129,10239282}. This thus renders the communication users/sensing targets more likely to be located in the near-field region of XL-arrays featured by spherical wavefronts~\cite{9913211,10185619,11015468}.
	Such a distinct near-field property not only provides new capabilities like \emph{spotlight beam-focusing} for improving communication performance, but also enables various applications such as high-accuracy localization, near-field wireless power transfer, and physical layer security~\cite{10149471,10845869,10827219,11185333}. 
	
	In this paper, we study near-field localization with XL-arrays. 
	In particular, unlike existing works that mostly assumed ideal hardware models, we consider hardware impairments (HIs) in near-field target localization, for which phase shifters (PSs) of XL-arrays are susceptible to HIs that introduce unknown (but static) phase biases~\cite{7080890}. For this case, existing block coordinate descent (BCD) based localization methods tend to get stuck in a low-quality solution when targets are close to the XL-array due to strong parameter coupling~\cite{10017173,10384355}. To address these issues, we propose a new three-phase HI-aware near-field localization method, which effectively detects faulty antennas and estimates positions of targets.

	\vspace{-8pt}
	\subsection{Related Works} \vspace{-2pt}
	\subsubsection{Target localization}
	The fundamental problem of far-field target direction-of-arrival (DoA) estimation has been a cornerstone of research in array signal processing for decades. 
	Early literature extensively explored a variety of methods~\cite{dai2025tutorial}, such as conventional beamforming (CBF), multiple signal classification (MUSIC), and estimating signal parameters via rotational invariance techniques (ESPRIT). 
	Although CBF is computationally straightforward,
	its localization resolution is limited when the targets are close to each other or the signal-to-noise ratio (SNR) is relatively low.
	In contrast, subspace-based methods like MUSIC and ESPRIT leverage the eigenstructure of data covariance matrix to achieve super angular resolution, hence enabling precise angle estimation even in complex environments.

	On the other hand, for near-field localization, the positions (including angle and range information) of targets can be effectively estimated at XL-arrays by exploiting  spherical wavefronts, without the need for distributed arrays and their synchronization.
	For example, the authors in~\cite{86917} proposed a two-dimensional (2D) MUSIC algorithm based on orthogonality between the signal and noise subspaces in near-field localization. 
	To reduce the complexity of 2D exhaustive search, a low-complexity yet efficient reduced-dimension (RD) MUSIC algorithm was developed in~\cite{8359308}, which transforms the 2D search into a one-dimensional (1D) local search.~In addition, a practical mixed-field target localization was considered in~\cite{zhou2025mixednearfieldfarfieldtarget}, where the authors proposed an efficient method for distinguishing far-field and near-field targets, as well as determining their angles and ranges (for near-field targets).
	Moreover, by leveraging the estimated positions of near-field scatterers, the authors in~\cite{10620236} proposed a novel method to~achieve simultaneous environment sensing and target localization even in the absence of a line-of-sight (LoS) path.


	
	\subsubsection{Hardware impairments}
	The above works on far/near-field target localization have mostly assumed ideal hardware models. Nevertheless, a variety of HIs may occur in practice, which can be largely classified as dynamic or static ones~\cite{8972463}. 
	Specifically, dynamic HIs introduce time-varying phase and/or amplitude errors due to e.g., discrete quantization errors, low resolution analog-to-digital converter (ADC)/digital-to-analog converter (DAC), additive distortion noises, and bit flip errors~\cite{8972463,7080890}.
	On the other hand, static HIs are characterized by unknown but static biases. For example, fixed phase biases may arise from manufacturing defects of PSs, disconnections caused by aging, or defects in ADC/DAC circuits~\cite{8972463,8248776}.
	Such fixed biases make the received signals deviate from the ideal signal where all antennas are functioning well, resulting in degraded localization performance in general.
	To address this issue, various far-field target localization methods accounting for HIs have been proposed in the literature. 
	For instance, the authors in~\cite{8932487} designed an efficient faulty antenna diagnosis method for correcting phase errors induced by HIs in uniform linear arrays (ULAs) and reconstructed the covariance matrix for enabling accurate angle estimation.
	This method was further extended in~\cite{9709659}, where the authors considered target localization for intelligent reflecting surface (IRS) systems, and detected faulty reflecting elements by solving a Toeplitz matrix reconstruction problem.
	In~\cite{10557757}, a more challenging scenario was considered, where some IRS elements were completely faulty (i.e., inability to receive or reflect signals). 
	To mitigate the detrimental effects of HIs, they employed transfer learning techniques to detect faulty elements and then reconstructed missing information from detected faulty elements, hence achieving highly accurate~localization.
	
	
	However, compared to traditional multiple-input multiple-output (MIMO) systems, the number of antennas at XL-arrays is significantly increased, rendering more antennas potentially operating in faulty states. 
	More seriously, the angle and range parameters of near-field targets are highly coupled in the received signals, hence resulting in degraded localization performance if not properly addressed~\cite{ghazalian2024calibrationrisaidedintegratedsensing}.
	To tackle this difficulty, the authors in~\cite{10017173,10384355} proposed iterative algorithms based on the BCD method to estimate faulty antennas, channel gains, angles and ranges alternately.
	For example, a Jacobi-Anger expansion based method was developed in~\cite{10017173} to decouple the angle and range dimensions, based on which a maximum likelihood problem was designed to detect faulty antennas and estimate the angles and ranges of targets.
	Besides, the authors in~\cite{10384355} designed a hybrid maximum likelihood and maximum a-posteriori estimator, which involves joint estimation of the positions of targets and phase biases of faulty antennas.

	

	\vspace{-6pt}
	\subsection{Motivations and Contributions}
	In view of the above works, there still exist several limitations as listed below. 
	\begin{itemize}
		\item \textbf{(Effects of HIs on near-field localization)}
		For near-field target localization, most existing works, such as~\cite{86917,8359308,zhou2025mixednearfieldfarfieldtarget}, made the assumption of ideal hardware models. 
		However, in practical XL-array systems, there may exist some antennas operating in faulty states due to HIs, which introduce unknown phase biases and hence degraded localization performance.  
		However, it is still unknown how the near-field localization performance is affected by HI parameters (e.g., fault probability, number of antennas) and  target position (i.e., angle and range).
		\item \textbf{(Incorrect localization of BCD-based methods)}
		The existing BCD-based methods in~\cite{10017173,10384355} may not be valid for the entire near-field region.
		In particular, when targets are close to the XL-array (but still beyond the Fresnel distance), small angle errors may incur large range errors due to the strong coupling between angle and range parameters in the channel model (see Section~\ref{Sec3}).
		Consequently, the BCD-based method may not provide accurate near-field localization.
	\end{itemize}
	
	
	Motivated by the above discussions, we consider in this paper a near-field multi-target localization system as shown in Fig.~\ref{fig:systemmodel}, where the targets actively send probing signals to facilitate localization at an XL-array base station (BS). 
	In particular, we assume that the BS is equipped with a ULA consisting of a number of faulty antennas with unknown (but static) phase biases.
	We aim to detect faulty antennas, as well as estimate the channel gain and the positions of near-field targets.
	The main contributions are summarized as follows.



	First, we show the limitations of existing BCD-based methods.
	Due to the strong coupling of angle and range parameters inherent in the near-field channel model, BCD-based methods may easily get stuck in  low-quality solutions, hence severely degrading localization accuracy, especially when targets are close to the XL-array.	
	To address this issue, we propose a new and efficient HI-aware near-field localization algorithm, which comprises three main phases.
	In Phase~1, faulty antennas are effectively detected by an efficient compressed sensing (CS) method, which is transformed into solving a sparse least-squares (LS) problem.
	To improve localization accuracy, we estimate coarse positions of targets by fusing multiple estimated angles of targets obtained from subarrays.
	In Phase~2, to correct phase errors induced by detected faulty antennas, we design a dedicated phase calibration method by temporally averaging the difference between the ideal fault-free phases and the actual phases with faulty antennas.
	Subsequently, in Phase~3, an efficient near-field localization method is devised to accurately estimate the positions of targets based on the full XL-array with phase calibration.
	In addition, we provide rigorous theoretical analysis for near-field localization under HIs by establishing the misspecified Cramér-Rao bound (MCRB) and the standard CRB, which quantify the localization performance in scenarios where HIs are present but unknown and in ideal scenarios where perfect knowledge of HIs is known \emph{a priori}, respectively. 
	Last, extensive numerical results are presented to demonstrate the effectiveness of our proposed algorithm for faulty-antenna detection and near-field target localization.

	\emph{Organization:} The remainder of this paper is organized as follows: Section~\ref{Sec2:label} introduces the near-field localization system with HIs. Section~\ref{Sec3} presents the problem formulation for faulty-antenna detection, as well as channel gains estimation and positions of targets, and points out the main limitations of existing methods.
	The proposed three-phase algorithm is presented in Section~\ref{Sec4}, with near-field localization performance analysis characterized in Section~\ref{Sec5}.
	Numerical results are presented in Section~\ref{Sec:SR}, with conclusions provided in Section~\ref{Sec:Con}.
	
	\emph{Notations:}  Lower-case, upper-case boldface letters denote vectors, matrices, respectively.  Upper-case calligraphic letters denote discrete and finite sets.
	$\mathcal{CN}(\mu,\sigma^2)$ represents the circularly symmetric complex Gaussian distribution with mean $\mu$ and variance $\sigma^2$.  For a vector or matrix, the superscripts $(\cdot)^{T}$ and $(\cdot)^{H}$ denote the transpose and Hermitian transpose, respectively. $\odot$ denotes the Hadamard product. For any two matrices $\mathbf{X}$ and $\mathbf{Y}$, $\mathbf{X} \succeq \mathbf{Y}$ means that the matrix $\mathbf{X}-\mathbf{Y}$ is positive semidefinite. $[\mathbf{X}]_{i,j}$ denotes the $i$-th row, the $j$-th column entry of $\mathbf{X}$. $\diag(\mathbf{x})$ outputs a diagonal matrix with the elements of a vector $\mathbf{x}$ on the diagonals.

	\vspace{-4pt}
	\section{System Model}\label{Sec2:label}\vspace{-2pt}
	We consider a multi-target localization system as shown in Fig.~\ref{fig:systemmodel}, where a BS equipped with an XL-array is employed to localize $K$ targets.
	\begin{figure}
		\centering
		\includegraphics[width=0.75\linewidth]{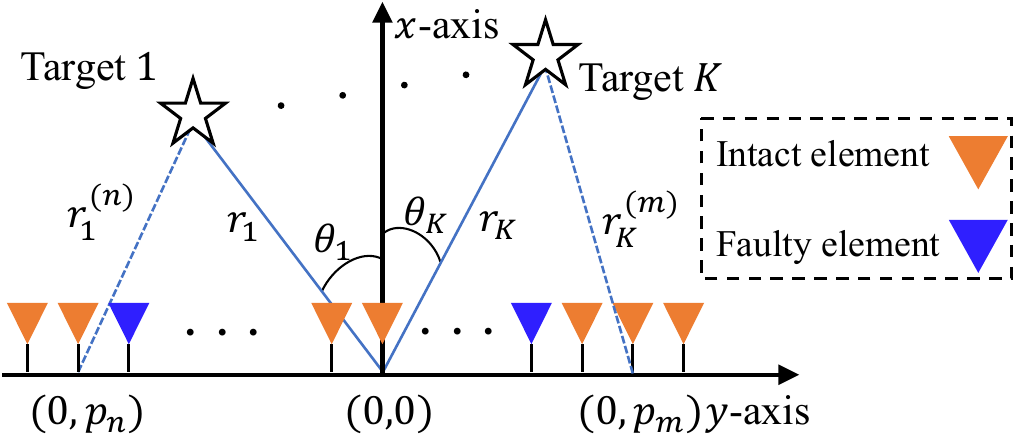}\vspace{-2pt}
		\caption{\centering The near-field target localization system under HIs.}
		\label{fig:systemmodel} \vspace{-14pt}
	\end{figure}
	\vspace{-8pt}
	\subsection{Channel Model}\vspace{-2pt}
	Without loss of generality, the XL-array with $N$ antennas is placed along the $y$-axis and the $n$-th antenna is located at $\left(0,l_n\right)$, where $l_n = \frac{2n-N-1}{2}d, \forall  n\in \mathcal{N} \triangleq \{ 1 ,\cdots,N \}$ denotes the location of antenna $n$, and $d = \lambda/2$ denotes the half-wavelength inter-antenna spacing.
	We assume that all targets are located in the near-field region of the XL-array, where the range between each target $k$ and the BS, denoted by $r_{k}, \forall k \in \mathcal{K}\triangleq \{1,2,\cdots,K\}$, is larger than the Fresnel distance $Z_{\text{Fres}} = 0.5\sqrt{\frac{D^3}{\lambda}}$ and smaller than the Rayleigh distance $Z_{\text{Rayl}} = 2D^2/\lambda$, with $D=(N-1)d \approx Nd$ denoting the array aperture~\cite{10778649}.
	We consider near-field localization scenarios in high-frequency bands, where the power of non-line-of-sight (NLoS) paths can be ignored due to severe path-loss and shadowing~\cite{10500334}.
	As such, we focus on the LoS channel scenario, for which the near-field channel steering vector $\mathbf{a}\left(\theta_k,r_k\right)  \in  \mathbb{C}^{N\times 1}$ can be modeled as \vspace{-2pt}
	\begin{align} \label{NF_vector_1}\vspace{-5pt}
		\mathbf{a}\left(\theta_k,r_k\right)  =  \Big[e^{-\frac{\jmath 2\pi}{\lambda} \left(r_k^{(1)}-r_k\right)},\cdots,e^{-\frac{\jmath 2\pi}{\lambda} \left(r_k^{(N)}-r_k\right)}\Big]^T,  \vspace{-2pt}
	\end{align}
	where $r_k^{(n)} = \sqrt{r_k^2 + l_n^2 -2 l_n r_k \sin(\theta_k)}$ denotes the range between the $k$-th target and the $n$-th antenna, and $\theta_k \in [-\frac{\pi}{2},\frac{\pi}{2}]$ denotes the physical DoA from the target to XL-array center.
	Then, by using the second-order Taylor approximation $\sqrt{1+x} \approx 1+\frac{1}{2}x-\frac{1}{8}x^2$, $r_{k}^{(n)}$ can be approximated as $r_k^{(n)} \approx r_k-l_n \sin(\theta_k) + \frac{l_n^2  \cos^2\left(\theta_k\right)}{2r_k}$, which is accurate enough when the target range is larger than $Z_{\text{Fres}}$.
	As such, the $n$-th entry in~\eqref{NF_vector_1} can be approximated as  \vspace{-2pt}
	\begin{align}
		\!\![\mathbf{a}\left(\theta_k,r_k\right)]_n  \approx \exp\Big({\frac{\jmath 2\pi}{\lambda} \big( l_n \sin(\theta_k)\!-\! \frac{l_n^2  \cos^2(\theta_k)}{2r_k} \big)}\Big).\!
	\end{align} 
	Based on the above, the near-field channel model from the $k$-th target to the XL-array can be modeled as 
	\begin{align} \label{NF_channel}
		\mathbf{h}_k =  \beta_k \mathbf{a}\left(\theta_k,r_k\right), 
	\end{align}
	where $\beta_k$ denotes the complex-valued channel gain.

	\vspace{-8pt}
	\subsection{Hardware Impairment Model} 
	In this paper, we focus on the static HI model caused by e.g., defects in manufacture and disconnections, which is \emph{a priori} unknown\footnote{ For the case with time-varying phase drifts, the static assumption no longer holds, resulting in cumulative phase errors that may significantly degrade localization performance.
		To address this challenge, adaptive tracking algorithms such as Kalman filtering~\cite{10012358} and posteriori Bayesian learning methods~\cite{10314450} can be employed to effectively estimate and track time-varying phase drifts.
	}~\cite{8248776}.
	Let $\mathbf{c} = [c_1,c_2,\cdots,c_N]^{T} \in \mathbb{C}^{N\times 1}$ denote the unknown HI coefficient vector of the XL-array. Herein, $c_n$, $\forall n \in \mathcal{N}$, characterizes the HI coefficient of antenna $n$:
	\begin{align} \label{HI}
		c_n =  \left\{\begin{array}{ll} e^{\jmath\zeta_n},&\text{if the $n$-th antenna is faulty}, \\1,&\text{if the $n$-th antenna is intact},\end{array}\right.
	\end{align}
	where $\zeta_n$ denotes random phase biases of a faulty antenna that is uniformly distributed within the region of $(0,2\pi)$.
	
	Without loss of generality, we assume that all antennas have the same likelihood to be in the faulty state, denoted by $p_{\text{fault}}$, and the faulty/intact states of different antennas are independent of each other\footnote{ In practical applications, clustered faults may happen, which leads to a block-sparse structure of the fault mask vector and hence renders the independent sparsity assumption invalid~\cite{8248776}.    
	For this case, specialized block-sparse recovery algorithms, such as block orthogonal matching pursuit (BOMP)~\cite{5424069} and block sparse Bayesian learning~\cite{9631375} can be employed to effectively detect clustered antenna faults.}.
	This stochastic behavior can be modeled by associating a Bernoulli random variable for each antenna, $\varrho_n \sim \mathrm{Ber}(p_{\text{fault}}), \forall n\in\mathcal{N}$.
	As such, the HI coefficient of antenna $n$ in~\eqref{HI} can be re-expressed as
	$	c_n = \varrho_n e^{\jmath\zeta_n}+1-\varrho_n$.
	When $\varrho_n =1$, the $n$-th antenna is faulty with $c_n = e^{\jmath\zeta_n}$; while when $\varrho_n =0$, the $n$-th antenna is intact with $c_n=1$.
	Therefore, given a fixed fault probability $p_{\text{fault}}$, as the total number of antennas $N$ increases, the average number of faulty antennas will increase proportionally and may become very large even for a small $p_{\text{fault}}$.

	\vspace{-8pt}
	\subsection{Signal Model}
	Let $s_{k,t}$ denote the probing signal of the $k$-th target at time index $t\in \mathcal{T} \triangleq \{1,2,\cdots,T_0\}$, where $T_0$ denotes the number of snapshots. We assume that all the $K$ targets send probing signals over $T_0$~\cite{zhou2025mixednearfieldfarfieldtarget}.
	As such,  the received uplink signal vector at the BS is given by\footnote{The obtained results can be extended to the case of passive sensing based on echo signals, which is detailed in Section~\ref{activelocaliza}.} \vspace{-2pt}
	\begin{align} \label{received_signal}
		\mathbf{y}_t &=  \sum_{k=1}^{K}  \beta_k \mathbf{c} \odot \mathbf{a}(\theta_k,r_k)s_{k,t} + \mathbf{n}_t \nonumber\\
		&= \diag(\mathbf{c})\mathbf{A}(\boldsymbol{\theta},\mathbf{r})\diag(\boldsymbol{\beta})\mathbf{s}_t + \mathbf{n}_t, \forall t\in\mathcal{T},\vspace{-2pt}
	\end{align} 
	where $\mathbf{A}\left(\boldsymbol{\theta},\mathbf{r}\right) \triangleq [ \mathbf{a}\left(\theta_1,r_1\right),\cdots,\mathbf{a}\left(\theta_K,r_K\right)] \in \mathbb{C}^{N\times K}$ denotes the near-field channel steering matrix, $\boldsymbol{\beta} = [\beta_1,\dots,\beta_K]^T\in\mathbb{C}^{K\times 1}$ denotes the complex-valued channel gain vector of $K$ targets, $\mathbf{s}_t =  [s_{1,t},\dots,s_{K,t}]^{T}\in\mathbb{C}^{K\times 1}$ denotes the signal vector of $K$ targets at time index $t$ and $\mathbf{n}_t \sim \mathcal{CN}(0,\sigma^2 \mathbf{I})$ is the received circularly symmetric complex Gaussian (CSCG) noise vector at time index $t$ with zero mean and variance $\sigma^2$.
	Moreover, we denote by $p_k$ the transmit power of the $k$-th source target signal, i.e., $\mathbb{E}_t\{|s_{k,t}|^2\} = p_k$.

	\vspace{-4pt}
	\section{Problem Formulation and Analysis}
	In this section, we first formulate an optimization problem for faulty antenna detection, as well as channel gain estimation and target localization. Then, we point out the limitations of existing BCD-based method when applied in near-field scenarios.
	
	\vspace{-8pt}
	\subsection{Problem Formulation}
	Based on the received signals $\mathbf{y}_t, \forall t\in \mathcal{T}$ in~\eqref{received_signal}, we aim to detect the HI coefficient vector $\mathbf{c}$, as well as estimate the channel gains $\boldsymbol{\beta}$ and positions of targets $(\theta_k,r_k),\forall k \in \mathcal{K}$. This problem can be mathematically formulated as 
	\begin{subequations}
		\begin{align}
			\textbf{(P1):} \min_{\mathbf{c},\boldsymbol{\beta},\boldsymbol{\theta},\mathbf{r}} &\quad \sum_{t=1}^{T_0}\| \mathbf{y}_t  -  \diag(\mathbf{c})\mathbf{A}(\boldsymbol{\theta},\mathbf{r})\diag(\boldsymbol{\beta})\mathbf{s}_t\|^2_2\nonumber\\
			\text{s.t.} &\quad~\eqref{HI}. \nonumber
		\end{align}
	\end{subequations}
	
	Problem \textbf{(P1)} is a non-convex optimization problem due to the intricate coupling of variables $\mathbf{c}$, $\boldsymbol{\beta}$, $\boldsymbol{\theta}$, and $\mathbf{r}$.
	Additionally, in near-field localization systems with XL-arrays, the number of antennas $N$ is extremely large (i.e., $N \gg T_0$), which makes Problem \textbf{(P1)} become an under-determined LS problem~\cite{9614042}.
	In general, the fault probability $p_{\text{fault}}$ is very small ~\cite{10339667}, thus the number of faulty antennas is much less than the total number of antennas.
	This motivates us to propose an effective CS-based faulty-antenna detection, as well as channel gain estimation and target localization method.

	Specifically, as most entries in $\mathbf{c}$ are equal to $1$ due to the faulty-antenna sparsity, the fault mask vector, defined as $\mathbf{z} \triangleq \mathbf{c} - \mathbf{1}$, is a sparse vector. 
	As such, Problem \textbf{(P1)} can be transformed into the following problem\vspace{-2pt}
	\begin{subequations}
		\begin{align} 
			\textbf{(P2):}\min_{\mathbf{c},\boldsymbol{\beta},\boldsymbol{\theta},\mathbf{r}} &\quad \sum_{t=1}^{T_0}\| \mathbf{y}_t  -  \diag(\mathbf{c})\mathbf{A}(\boldsymbol{\theta},\mathbf{r})\diag(\boldsymbol{\beta})\mathbf{s}_t\|^2_2 + \rho \| \mathbf{z} \|_1  \nonumber\\ 
			\text{s.t.} & \quad \mathbf{z} =  \mathbf{c} - \mathbf{1}, 
		\end{align}
	\end{subequations}
	where $\ell_2$ norm characterizes the approximation error between the received signals and the reconstructed signals, $\ell_1$ norm represents the penalty on sparsity of the fault mask vector, and $\rho$ denotes the regularization parameter to balance the approximation error and sparsity penalty.


	\vspace{-9pt}
	\subsection{Existing BCD-based Method} \label{Sec3}
	To solve Problem \textbf{(P2)}, one possible solution method is alternately optimizing the HI coefficient vector, the channel gains, as well as the angles and ranges of targets by using the BCD-based technique (see, e.g.,~\cite{10017173,10384355}).
	Specifically, this method estimates $\mathbf{c}$, $\boldsymbol{\beta}$, $\boldsymbol{\theta}$, and $\mathbf{r}$ in an alternating manner, as detailed in the following.
	
	\subsubsection{\underline{\textbf{Optimize fault mask vector}}} \label{sec:up_mask}
	Given any (estimated) $\boldsymbol{\beta}$, $\boldsymbol{\theta}$ and $\mathbf{r}$, Problem \textbf{(P2)} reduces to the fault mask vector optimization problem as follows\vspace{-2pt}
	\begin{equation}
		\begin{aligned} \label{problem_2}
			\textbf{(P3):}\min_{\mathbf{z}} \; \underbrace{\sum_{t=1}^{T_0}\| \mathbf{y}_t \! -\!  \diag(\mathbf{z}\!+\!\mathbf{1})\mathbf{A}(\boldsymbol{\theta},\mathbf{r})\diag(\boldsymbol{\beta})\mathbf{s}_t\|^2_2 }_{f(\mathbf{z})} \!+\! \underbrace{\rho \| \mathbf{z} \|_1}_{g(\mathbf{z})},\nonumber
		\end{aligned}  \vspace{-2pt}
	\end{equation}
	where $f(\mathbf{z})$ is a smooth convex function and $g(\mathbf{z})$ represents the non-smooth regularization term. 
	This problem can be solved by using the iterative shrinkage-thresholding algorithm (ISTA)~\cite{080716542}, which will be elaborated in Section~\ref{sec_fault_detec}.

	\subsubsection{\underline{\textbf{Estimate channel gain and positions}}} Given any mask vector $\hat{\mathbf{z}}$ obtained by solving Problem \textbf{(P3)},  the HI coefficient~vector can be expressed as $\hat{\mathbf{c}}=\hat{\mathbf{z}}+\mathbf{1}$. 
	Let $\boldsymbol{\Phi}_t(\boldsymbol{\theta},\mathbf{r}) \triangleq \diag(\hat{\mathbf{c}})\mathbf{A}({\boldsymbol{\theta}},{\mathbf{r}})\diag (\mathbf{s}_t)$, Problem \textbf{(P2)} reduces to the following problem\vspace{-4pt}
	\begin{align}\vspace{-2pt}
		\textbf{(P4):} \min_{\boldsymbol{\beta},\boldsymbol{\theta},\mathbf{r}} \quad \sum_{t=1}^{T_0}\| \mathbf{y}_t  - \boldsymbol{\Phi}_t(\boldsymbol{\theta},\mathbf{r})\boldsymbol{\beta}\|^2_2. \nonumber
	\end{align}
	The optimal $\boldsymbol{\beta}$ for any given $({\boldsymbol{\theta}},{\mathbf{r}})$ can be obtained as 
	\vspace{-2pt}
	\begin{align}\label{beta_solved}
		\hat{\boldsymbol{\beta}} = 
		\boldsymbol{\Phi}^{\dag}({\boldsymbol{\theta}},{\mathbf{r}})\mathbf{y},
	\end{align}
	where  $\boldsymbol{\Phi}({\boldsymbol{\theta}},{\mathbf{r}}) = [\boldsymbol{\Phi}^T_t(\boldsymbol{\theta},\mathbf{r}),\dots,\boldsymbol{\Phi}^T_{T_0}(\boldsymbol{\theta},\mathbf{r})]^T$,  $\boldsymbol{\Phi}^{\dag}({\boldsymbol{\theta}},{\mathbf{r}}) = 
	( \boldsymbol{\Phi}^H({\boldsymbol{\theta}},{\mathbf{r}}) \boldsymbol{\Phi}({\boldsymbol{\theta}},{\mathbf{r}}) )^{-1} \boldsymbol{\Phi}^H({\boldsymbol{\theta}},{\mathbf{r}})$ denotes the pseudo inverse of $\boldsymbol{\Phi}({\boldsymbol{\theta}},{\mathbf{r}})$, and $\mathbf{y} = [\mathbf{y}^T_1,\dots,\mathbf{y}^T_{T_0}]^T$.
	Then, by substituting \eqref{beta_solved} into Problem \textbf{(P4)}, and denoting $\mathbf{\Pi}_{\boldsymbol{\Phi}({\boldsymbol{\theta}},{{\mathbf{r}}})}^{\bot}= \mathbf{I}-\boldsymbol{\Phi}({\boldsymbol{\theta}},{{\mathbf{r}}})\boldsymbol{\Phi}^{\dag}(\hat{\boldsymbol{\theta}},{{\mathbf{r}}})$ as the orthogonal projection matrix on the column space of $\boldsymbol{\Phi}({\boldsymbol{\theta}},{{\mathbf{r}}})$, we can solve the problem by updating $\boldsymbol{\theta}$ and $\mathbf{r}$ alternately as follows.

	\begin{itemize}
		\item \textit{Update $\boldsymbol{\theta}$:} 
		Given $\hat{\mathbf{r}}$, Problem \textbf{(P4)} reduces to the following problem for angle estimation of targets \vspace{-2pt}
		\begin{align}
			\textbf{(P5):}  \min_{{\theta}_k} \quad \| \mathbf{\Pi}_{\boldsymbol{\Phi}({\boldsymbol{\theta}},{\hat{\mathbf{r}}})}^{\bot}\mathbf{y} \|^2_2,\forall k\in \mathcal{K}. \nonumber 
		\end{align}
		\item \textit{Update $\mathbf{r}$:} Given   $\hat{\boldsymbol{\theta}}$, Problem \textbf{(P4)} reduces to the following problem for range estimation of targets \vspace{-2pt}
		\begin{align} 
			\textbf{(P6):} \min_{{r}_k} \quad \| \mathbf{\Pi}_{\boldsymbol{\Phi}(\hat{\boldsymbol{\theta}},{{\mathbf{r}}})}^{\bot}\mathbf{y} \|^2_2,\forall k\in \mathcal{K}. \nonumber
		\end{align} 
	\end{itemize}
	\vspace{-2pt}


	To solve Problems \textbf{(P5)} and \textbf{(P6)}, a line-search method can be respectively applied in the angle domain within $\theta_k\in[-\frac{\pi}{2},\frac{\pi}{2}]$ and the range domain within $r_k \in [Z_{\text{Fres}}, Z_{\text{Rayl}}]$ for finding angle and range solutions.
	{
	However, the objective function of Problem \textbf{(P4)} exhibits strong coupling between the angle and range parameters, especially when the target range is small (see Fig.~\ref*{fig:lscostangle0416}). This coupling makes the BCD-based algorithm, which alternately updates angle and range parameters, highly susceptible to getting \emph{stuck} in a {low-quality}~solution. In particular, a slight deviation in angle estimation can lead to significant errors in range estimation, making the algorithm converge to incorrect solutions. 
	For clarity, we provide Example \ref{ex1} to illustrate this issue.
	}
	This issue will be effectively addressed in the next~section. \vspace{-2pt}
	\begin{example}\label{ex1}\emph{
			We illustrate this issue with a numerical example in Fig.~\ref*{fig:lscostdis0416}, which shows the objective function of Problem \textbf{(P4)} versus the target range. 
			One can observe that when the target range is small (e.g., $15.5$ m), the target range estimated by the BCD-based method tends to converge to a wrong value, i.e., $33.6$ m (estimated range) versus $15.5$ m (true range), although there is only a slight error in the estimation of target angle, i.e., $-0.009$ rad (estimated angle) versus $0$ rad (true angle), thus failing to accurately localize near-field target.
		}\vspace{-5pt}
	\end{example}
	
	\begin{figure}[t]
		\centering
		\subfigure[Objective function of Problem \textbf{(P2)} versus angle and range of target.]{\includegraphics[width=0.47\linewidth]{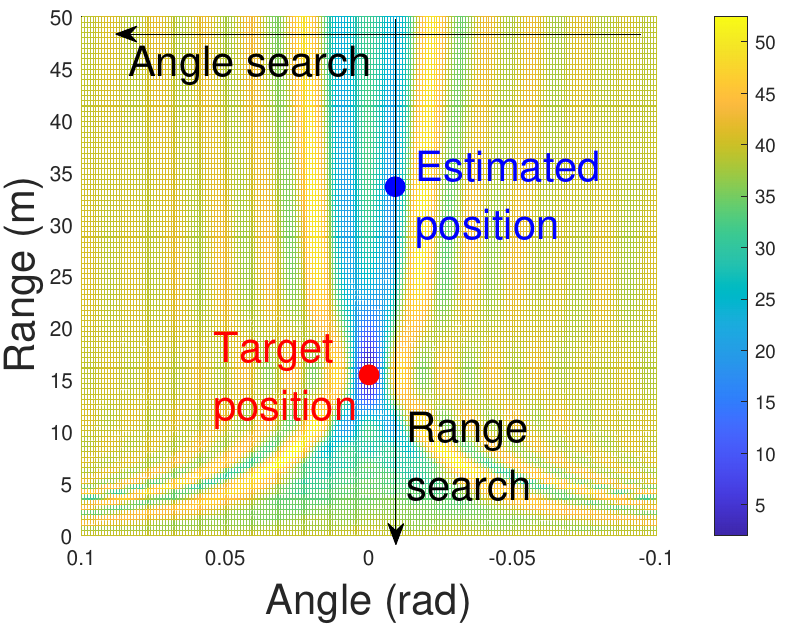}
			\label{fig:lscostangle0416}} \hfill
		\subfigure[Objective function of Problem \textbf{(P2)} versus range of target.]{	\includegraphics[width=0.45\linewidth]{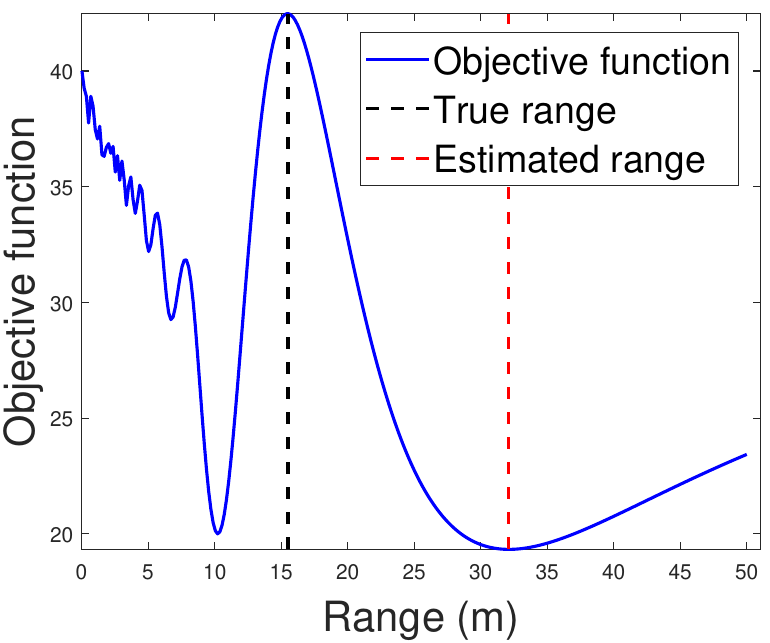}
			\label{fig:lscostdis0416}}\hspace{-4pt}
			\vspace{-5pt}
		\caption{The objective function of Problem \textbf{(P2)} for the conventional BCD-based method.
			System parameters are set as $N = 256$, $\lambda = 0.01\text{ m}$, and $p_{\text{fault}} = 2\%$.
			The position of the target is $(0 \text{ rad}, 15.5\text{ m})$, and the estimated position is $(-0.009 \text{ rad}, 33.6 \text{ m})$.} \label{fig:Cost} 	\vspace{-14pt}
	\end{figure}

	\vspace{-6pt}
	\section{Proposed Localization Algorithm} \vspace{-2pt} \label{Sec4}
	In this section, we propose an efficient three-phase target localization method that sequentially detects faulty antennas, calibrates phase bias, and then performs fine-grained target localization.
	The framework of the proposed method is illustrated in Fig.~\ref{fig:framework}, with main procedures given as follows.
	\begin{itemize}
		\item \textbf{Phase~1 (fault detection):} Based on the received signals $\mathbf{y}_t$ in~\eqref{received_signal}, we first detect faulty antennas by alternately optimizing the HI coefficients, channel gains, and estimating positions of targets. 
		Specifically, given estimated positions of targets, we detect faulty antennas by using the gradient descent method with a shrinkage operator.
		Then, given fixed detected faulty antennas, we partition the XL-array into multiple subarrays to eliminate the impact of HIs on near-field target localization.
		As such, coarse target localization can be achieved by fusing the information of estimated angles from all subarrays.
		Besides, the channel gains are estimated by solving an LS problem based on the coarse estimation.
		Note that in this phase, the coarse target localization is used for achieving accurate HI coefficient estimation, while the positions of targets will be refined in Phase~3. 
		\item \textbf{Phase~2 (phase calibration):} 
		Based on detected faulty antennas in Phase~1, effective phase calibration based on temporal 
		averaging is carried out in Phase~2, minimizing the difference between the observed and fault-free phases to recover true phases.
		\item \textbf{Phase~3 (target localization):}
		Based on the full XL-array with phase calibration in Phase~2, we then propose an efficient near-field fine-grained localization method.
		Specifically, the angle and range parameters are decoupled by reconstructing an effective covariance matrix based on calibrated received signals.
		Then, the angle and range MUSIC spectral searches are performed to estimate the positions of targets with high spatial resolution.
	\end{itemize}
	\begin{figure*}
		\centering
		\includegraphics[width=0.75\linewidth]{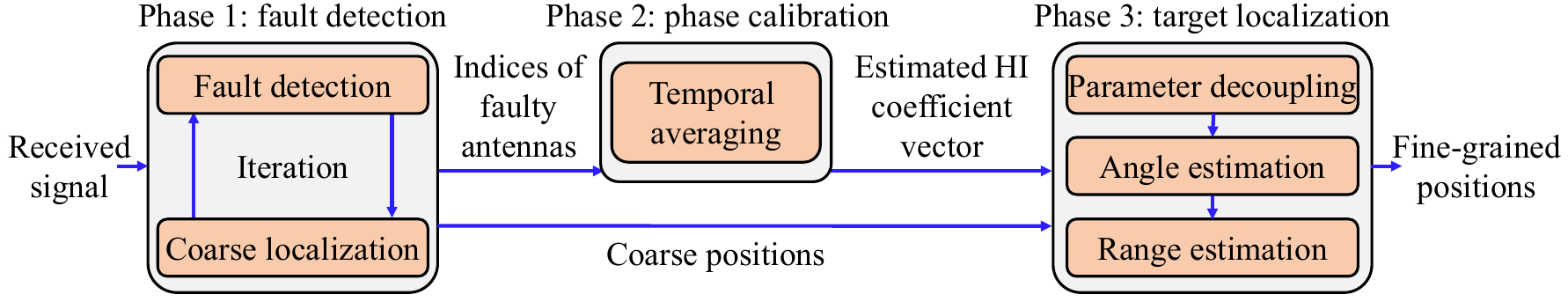}\vspace{-2pt}
		\caption{\centering The framework of proposed three-phase near-field target localization method.}
		\label{fig:framework} \vspace{-16pt}
	\end{figure*}
	\begin{remark}
		\emph{ (Effectiveness of proposed method).
			Note that compared with the existing BCD-based method in~\cite{10017173,10384355}, the proposed method decouples the fault detection and target localization in Phase~1 and Phase~3, respectively, while Phase~2 is used for phase calibration based on Phase~1 for enabling fine-grained localization in Phase~3. 
			As such, the proposed three-phase method can effectively tackle the issue of the existing BCD-based method, hence achieving significantly enhanced target localization accuracy, especially when the targets are close from the XL-array.
		}
	\end{remark}

	\vspace{-15pt}
	\subsection{Phase~1: Fault Detection} \vspace{-2pt} \label{Sec:fault_detection}
	In this phase, we detect faulty antennas by alternately optimizing HI coefficient vector, estimating the channel gain and positions of targets until convergence.
	Note that the goal of Phase~1 is to detect faulty antennas, while coarse target localization are performed to improve the detection accuracy.
	
	\subsubsection{\underline{\textbf{Fault detection}}} \label{sec_fault_detec}
	Given any (estimated) channel gain~$\boldsymbol{\beta}$ and positions $(\boldsymbol{\theta},\mathbf{r})$, faulty antennas can be detected by solving Problem (\textbf{P3}) in Section~\ref{sec:up_mask}.
	Specifically, Problem (\textbf{P3}) is an $\ell_1$-regularized LS problem, which can be solved by ISTA to obtain its optimal solution~\cite{080716542}.
	This method separates~the objective function into a differentiable term $f(\mathbf{z})$ and a non-differentiable $\ell_1$ regularization term $g(\mathbf{z})$, where gradient descent is performed on $f(\mathbf{z})$ to obtain an update direction, and the sparsity of regularization term is satisfied by the shrinkage operator.
	Mathematically, the shrinkage~operator~$S(\mathbf{z},\gamma)$~is~defined~as~\cite{080716542}
	\begin{equation}
		[S(\mathbf{z},\gamma)]_n = \text{sign}(z_n)(|z_n|-\gamma),
	\end{equation}
	where $[S(\mathbf{z},\gamma)]_n$ denotes the $n$-th entry of $S(\mathbf{z},\gamma)$, and $\gamma \ge 0$ is a constant.
	As such, based on the gradient and shrinkage operator, the iterative update for mask vector $\mathbf{z}$ is given by~\cite{080716542}
	\begin{equation} \label{eq:updata}
		\ddot{\mathbf{z}} = S(\mathbf{z} - \nu \nabla f(\mathbf{z}), \nu\rho),
	\end{equation}
	where $\nu$ denotes the update step size, $\ddot{\mathbf{z}}$ denotes the update of $\mathbf{z}$,
	$\nabla f(\mathbf{z})$ denotes the gradient of $f(\mathbf{z})$ given by
	\begin{equation} \label{gradient_f}
		\begin{aligned}
			\nabla f(\mathbf{z}) = - \sum_{t=1}^{T_0}\diag \left( \mathbf{x}_t\right)^H \left( \mathbf{y}_t  -  \diag(\mathbf{z}+\mathbf{1})\mathbf{x}_t \right),
		\end{aligned} 
	\end{equation}
	where $\mathbf{x}_t = \mathbf{A}(\boldsymbol{\theta,\mathbf{r}})\diag(\boldsymbol{\beta})\mathbf{s}_t$.
	The above iterations terminate when the update error is less than a preset threshold $\varsigma$, i.e., 
	$
	|\ddot{\mathbf{z}} - \mathbf{z}|\le \varsigma
	$.
	We denote $\hat{\mathbf{z}} = [\hat{z}_1,\dots,\hat{z}_N]^T$ as the obtained fault mask vector after convergence, in which non-zero entries indicate faulty antennas, we obtain the indices of faulty antennas as
	\begin{equation} 
		\begin{aligned} \label{set_estimated}
			\mathcal{D} \triangleq  \{n \mathcal{N} \Big|   |\hat{z}_n |>\tau\},
		\end{aligned} 
	\end{equation}
	where $\tau$ is a given threshold,\footnote{Due to noise, HIs, and limited measurements, the estimated $\hat{\mathbf{z}}$ may not be strictly sparse~\cite{9325942}. A simple and effective way to detect faulty antennas is comparing the estimated fault mask vector $\hat{\mathbf{z}}$ with a threshold $\tau$, which can be set based on the fault probability.} and $\mathcal{D}$ contains the indices of estimated faulty antennas.

	
	\subsubsection{\underline{\textbf{Coarse localization}}} \label{Sec:Coarse}
	
	
	Given any detected faulty antennas $\mathcal{D}$, we then estimate the positions of targets and channel gains.
	In order to fully utilize the received signals from \emph{fault-free} antennas while avoiding the influence of faulty antennas, we divide the XL-array into multiple subarrays for individual~angle estimation, and then estimate the positions of targets by fusing the information of angle estimation from all subarrays.
	In particular, we partition the entire array into $Q$ subarrays based on the following two criteria\footnote{ The sparsity of antenna faults indicates that there is likely to be several intact contiguous subarrays available for coarse localization after partitioning, which will be validated in subsequent numerical results.
	For scenarios with high fault probability where the sparsity assumption does not hold, alternative methods such as the matrix pencil method (MPM)~\cite{9783040} and deep learning-based approaches~\cite{10528239} can be employed to effectively diagnose antenna faults without relying on sparsity.}, as illustrated in Fig.~\ref{fig:arraypartition}:

	\begin{itemize}
		\item \textit{Partition criterion 1:} 
		Each subarray should not contain~any faulty antenna to avoid the impact of phase errors caused by faulty-antenna on localization. 
		As such, according to the detected faulty-antenna indices $\mathcal{D}$, the original XL-array is adaptively partitioned into $|\mathcal{D}|\!+\!1$ contiguous subarrays, each containing intact antennas~only. 
		
		\item \textit{Partition criterion 2:} All targets should be located in the far-field region of each subarray for enabling far-field based angle estimation.
		To this end, these $|\mathcal{D}|+1$ subarrays are further partitioned into $Q$ subarrays by adjusting the aperture of each subarray.
		Mathematically, the array partition ensures that all ranges between targets and each subarray should satisfy the far-field condition 
		\begin{align}
			r_{k,q} \geq {2D_q^2}/{\lambda},\forall k\in \mathcal{K},q\in \mathcal{Q} \triangleq \{1,2,\cdots,Q\},
		\end{align}
		where $r_{k,q}$ denotes the range between the $k$-th target and the $q$-th subarray, $D_q=N_qd$ and $N_q$ denote the aperture of the $q$-th subarray and the number of antennas in the $q$-th subarray, respectively.
		This is equivalent to satisfy the condition that the Rayleigh distance of each subarray is no larger than the Fresnel distance of the entire array, i.e., $Z_{\text{Rayl},q} = 2D_q^2/\lambda \le Z_{\text{Fres}}$, yielding $N_q \le  2^{-\frac{3}{4}}N^{\frac{3}{4}} \approx 0.6N^{\frac{3}{4}}\triangleq N_{\text{sub}}, \forall q\in \mathcal{Q}$, where $N_{\text{sub}}$ denotes the maximum number of antennas of subarrays.
		{ 
		Note that to successfully resolve $K$ targets, each subarray must contain at least $K$ antennas, i.e., $K \le N_q \le N_{\text{sub}}$.
		Thus, subarrays failing to meet this condition will be excluded from subsequent processing.}
	\end{itemize}
		\begin{figure}
		\centering
		\includegraphics[width=0.85\linewidth]{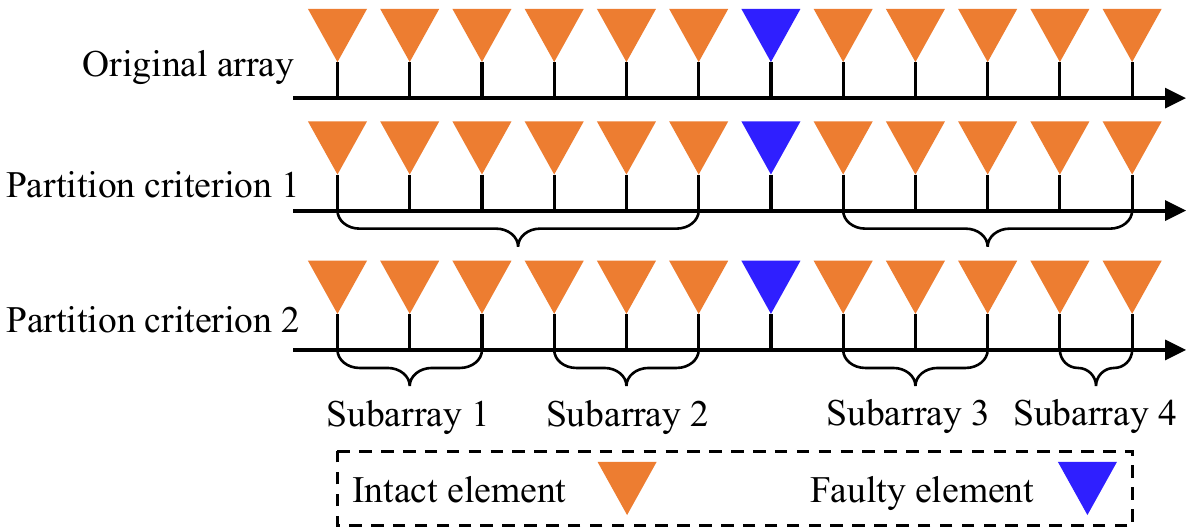}\vspace{-2pt}
		\caption{\centering Illustration of array partition.}
		\label{fig:arraypartition} \vspace{-18pt}
	\end{figure} 
	
	Then, we estimate angles of targets with respect to (w.r.t.) each subarray center by using conventional far-field localization techniques.
	Specifically, let $\mathbf{y}_{q,t}\in \mathbb{C}^{N_q\times 1}$ and $\mathcal{N}_q \subset \mathcal{N}$ denote the received signal vector and the antenna set of the $q$-th subarray, respectively. The covariance matrix of the received signals $\mathbf{y}_{q,t}$ over $T_0$ for the $q$-th subarray can be expressed as 
	\begin{align} \label{eq:covariance_mat12}
		\mathbf{R}_q &= \mathbb{E}_{t}\{\mathbf{y}_{q,t}\mathbf{y}^{H}_{q,t}\} = \mathbf{B}_q\mathbf{G}\mathbf{B}_q^H + \sigma^2 \mathbf{I}, \forall q\in \mathcal{Q},
	\end{align} 
	where $\mathbf{b}_q(\theta_{k,q}) \in \mathbb{C}^{N_q\times 1}$ denotes the far-field steering vector, given by $[\mathbf{b}_q(\theta_{k,q})]_n =  e^{-\jmath \frac{2\pi}{\lambda} l_n \sin(\theta_{k,q})},\forall n\in \mathcal{N}_q$, with $\theta_{k,q}$ denoting the angle from the $k$-th target to the $q$-th subarray.
	$\mathbf{B}_q \triangleq [\mathbf{b}_q(\theta_{1,q}),\cdots,\mathbf{b}_q(\theta_{K,q})] \in \mathbb{C}^{N_q \times K}$ denotes the steering matrix of the $q$-th subarray, and $\mathbf{G} \!=\! \diag\{ g_{1},\cdots,g_{K}\}$ with $g_{k}\!\!=\!\! |\beta_k|^2p_k$ is the equivalent received signal power matrix.
	In practice, the theoretical covariance matrix $\mathbf{R}_q$ can~be approximated by the sample covariance matrix $\overline{\mathbf{R}}_q = \frac{1}{T_0}\sum_{t=1}^{T_0}\mathbf{y}_{q,t}\mathbf{y}^{H}_{q,t} $.
	To estimate the angles of the $q$-th subarray, we perform the eigenvalue decomposition (EVD) of $\overline{\mathbf{R}}_q$ as
	\begin{align} \label{EVD_R_q}
		\!\!\!\overline{\mathbf{R}}_q \!&=\!  \mathbf{U}_q  \mathbf{V}_q  \mathbf{U}_q^H = \mathbf{U}_{\mathrm{S},q} \mathbf{V}_{\mathrm{S},q} \mathbf{U}_{\mathrm{S},q}^{H} + \mathbf{U}_{\mathrm{N},q} \mathbf{V}_{\mathrm{N},q} \mathbf{U}_{\mathrm{N},q}^{H} \nonumber\\
		&=\! [\mathbf{u}_1,\cdots\!,\mathbf{u}_{N_q}] \diag\{v_1,\cdots\!,v_{N_q}\}[\mathbf{u}_1,\cdots\!,\mathbf{u}_{N_q}]^H\!, \forall q,\!
	\end{align}  
	where $\mathbf{V}_q\in \mathbb{C}^{N_q\times N_q}$ is a diagonal matrix with the eigenvalues arranged as $|v_1|\ge\cdots\ge |v_{N_q}|$, and $\mathbf{U}_q ^{N_q\times N_q}$  is the corresponding eigenvector matrix.
	Herein, $\mathbf{U}_{\mathrm{S},q}\in \mathbb{C}^{N_q\times K}$ and $\mathbf{U}_{\mathrm{N},q}\in \mathbb{C}^{N_q\times (N_q-K)}$ denote the signal and noise subspaces, respectively, $\mathbf{V}_{\mathrm{S},q} = \diag\{v_1,\cdots,v_K\}\in \mathbb{C}^{K \times K}$ is a diagonal matrix consisting of the $K$ largest eigenvalues,
	and $\mathbf{V}_{\mathrm{N},q} = \diag\{v_{K+1},\cdots,v_{N_q}\}\in \mathbb{C}^{(N_q-K) \times (N_q-K)}$ is a diagonal matrix consisting of the $N_q-K$ smallest eigenvalues.
	\begin{remark}\emph{(Unknown number of targets). 
			The MUSIC method relies on the separation of the signal and noise subspaces from the covariance matrix in~\eqref{EVD_R_q}, which depends on prior knowledge of the number of targets $K$.
			Therefore, the MUSIC algorithm may not always perform well when the number of targets is unknown.
			In this case, under high-SNR conditions, the eigenvalues corresponding to the signal subspace are significantly larger than those of the noise subspace. 
			As such, after sorting all eigenvalues in a descending order, the number of targets $K$ can be determined by summing eigenvalues that account for a sufficiently high proportion (e.g., 95\%) of the total cumulative eigenvalues.
			This issue, however, may become serious in the low-SNR regime, since the eigenvalues of the signal subspace and noise subspace become roughly comparable, making it difficult to distinguish between them.
			In this case, deep learning-based methods, such as~\cite{9457195,10266765}, can be applied to directly learn positions of targets, thereby avoiding the need for acquiring prior knowledge of the number of targets.
		}
	\end{remark}
	Based on the EVD of the sample covariance matrix $\overline{\mathbf{R}}_q$, the angles of targets w.r.t. each subarray can be determined by identifying $K$ peaks in the MUSIC spectrum. 
	Specifically, since the noise subspace is orthogonal to the signal subspace,
	we have $\mathbf{U}^{H}_{\mathrm{N},q}\mathbf{B}_q = \mathbf{0}$, which can be rewritten as $\mathbf{b}^{H}_q(\theta_{k,q})\mathbf{U}_{\mathrm{N},q}\mathbf{U}^{H}_{\mathrm{N},q}\mathbf{b}_q(\theta_{k,q}) = 0,\forall k \in \mathcal{K}$.
	Therefore, the estimated angles from targets to the $q$-th subarray, denoted as $[\hat{\theta}_{1,q},\hat{\theta}_{2,q},\dots,\hat{\theta}_{K,q}]$ (in an ascending order), can be obtained by searching for the $K$ largest peaks in angle MUSIC spectrum:
	\begin{equation}
		\begin{aligned}\label{angle_spectrum}
			F_q(\theta) =  \left[\mathbf{b}^{H}_q(\theta)\mathbf{U}_{\mathrm{N},q}\mathbf{U}^{H}_{\mathrm{N},q}\mathbf{b}_q(\theta)\right]^{-1},\forall q\in \mathcal{Q}.
		\end{aligned} 
	\end{equation}
	
	Subsequently, we estimate the coarse positions of targets using the estimated angles of all subarrays (i.e., $\{\hat{\theta}_{k,q}\}$).
	The positions of targets can be estimated by triangulation~\cite{4343996}, which uses the geometric properties of triangles as shown in Fig.~\ref{fig:coarselocalization}.
	Specifically, let  $\mathbf{e}_{k,q} = [\cos(\hat{\theta}_{k,q}),\sin(\hat{\theta}_{k,q})]^{T}$ denote the unit direction vector pointing towards the $k$-th target from the $q$-th subarray.
	Then, coarse positions of targets can be obtained by solving the following problem
	\begin{equation}
		\begin{aligned}\label{Localization_subarray}
			\textbf{(P7):} \min_{\boldsymbol{\xi}_k}\quad  \sum_{q=1}^{Q}\left|\boldsymbol{\xi}_k-\left(\mathbf{p}_q + \mathbf{e}_{k,q}( \boldsymbol{\xi}_k-\mathbf{p}_q )^T \mathbf{e}_{k,q}\right)\right|^2,
		\end{aligned} 
	\end{equation}
	where $\mathbf{p}_q$ denotes the position of the $q$-th subarray center and $\boldsymbol{\xi}_k = [\xi_{k,1},\xi_{k,2}]^T$ denotes  the $k$-th target in the Cartesian coordinate system, respectively.
	\begin{figure}
		\centering
		\includegraphics[width=0.65\linewidth]{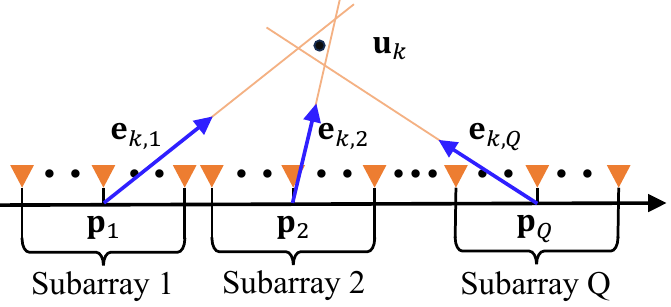} \vspace{-2pt}
		\caption{\centering Illustration of coarse localization.}
		\label{fig:coarselocalization} \vspace{-18pt}
	\end{figure}
	\begin{lemma} \label{Lemma1}
		\emph{
			The optimal solution to Problem \textbf{(P7)} is given by
			\begin{equation}
				\begin{aligned} \label{P7_solution}
					\hat{\boldsymbol{\xi}}_k = \Big(\sum_{q=1}^{Q} (\mathbf{I}-\mathbf{e}_{k,q} \mathbf{e}_{k,q}^T)  \Big)^{-1} \sum_{q=1}^{Q} (\mathbf{I}-\mathbf{e}_{k,q}\mathbf{e}_{k,q}^T) \mathbf{p}_q.
				\end{aligned} 
			\end{equation}
		}
	\end{lemma}
	\begin{proof}
		Please refer to Appendix~\ref{App1}.
	\end{proof}
	Based on \textbf{Lemma~\ref{Lemma1}}, we obtain the angle and range parameters by transforming the Cartesian coordinate into the polar one
	\begin{equation}
		\begin{aligned} \label{location_course_result}
			\hat{r}_k = \sqrt{\hat{\xi}^{2}_{k,1} + \hat{\xi}^{2}_{k,2}}, \quad \hat{\theta}_k   = \arcsin (\hat{\xi}_{k,2}/\hat{r}_k), \forall k\in \mathcal{K}.
		\end{aligned} 
	\end{equation}

	After obtaining the coarse estimation of angle and range, the problem to estimate complex channel gain $\boldsymbol{\beta}$ is formulated as
	\begin{align}
		\textbf{(P8):} \min_{\boldsymbol{\beta}} \quad \sum_{t=1}^{T_0}\| \mathbf{y}_t  - \boldsymbol{\Phi}_t(\hat{\boldsymbol{\theta}},\hat{\mathbf{r}})\boldsymbol{\beta}\|^2_2. \nonumber
	\end{align}
	Similar to \eqref{beta_solved}, the optimal solution to Problem \textbf{(P8)} can be obtained as
	\begin{align} \label{channel_gain_solution}
		\hat{\boldsymbol{\beta}} = 
		\boldsymbol{\Phi}^{\dag}(\hat{\boldsymbol{\theta}},\hat{\mathbf{r}})\mathbf{y}.
	\end{align}
	
	\subsection{Phase~2: Phase~Calibration} 
	After faulty-antenna detection, Phase~2 aims to calibrate phase biases induced by the detected faulty antennas, which will be used in Phase~3 for fine-grained localization. 
	Given detected faulty antennas $\forall i \in \mathcal{D}$ in~\eqref{set_estimated}, we propose a low-complexity phase biases estimation method based on the phase differences between the received signals and fault-free signals.
	
	Specifically, for an arbitrary faulty-antenna $\forall i \in \mathcal{D}$, the observed phase of received signal based on~\eqref{received_signal} is given by 
	\begin{equation}
		\begin{aligned}
			\!\!\!\phi^{\text{Obs}}_{i,t} \!= \! \Big(\!\frac{2\pi}{\lambda} (\hat{r}^{(i)}_{k} \!-\! \hat{r}_k) + \zeta_i + \varepsilon_{i,t}\!\Big)\!\mod 2\pi, \forall i\in \mathcal{D}, t\in \mathcal{T},\!\!
		\end{aligned} 
	\end{equation}
	where $\varepsilon_{i,t}$ denotes the phase noise term introduced by the noise of received signal.
	The ideal phase of fault-free received signal (without phase biases and noise) can be expressed as 
	\begin{equation}
		\begin{aligned}
			\phi^{\text{Ideal}}_{i,t} = \Big(\frac{2\pi}{\lambda} (\hat{r}^{(i)}_{k} - \hat{r}_k)\Big) \mod 2\pi, \forall t\in \mathcal{T},
		\end{aligned} 
	\end{equation}
	which is obtained from array geometry.
	As such, the phase difference of the $i$-th antenna is given by 
	\begin{equation}
		\Delta\phi_{i,t} = \left(\phi^{\text{Obs}}_{i,t} -\phi^{\text{Ideal}}_{i,t} \right) \mod 2\pi, \forall t\in \mathcal{T}. 
	\end{equation}
	To resolve phase wrapping ambiguity, we impose a phase continuity constraint through unwrapping~\cite{8672140}, i.e., 
	\begin{equation}
		\Delta\phi^{\text{dewarp}}_{i,t} = \Delta\phi_{i,t} + 2\pi \omega_{i,t}, \forall \omega_{i,t} \in\mathbb{Z}, \forall t\in \mathcal{T},
	\end{equation}
	where the integer multiple $\omega_{i,t}$  is selected to minimize temporal phase discontinuity, i.e. 
	\begin{equation}
		\omega_{i,t} = \arg\min_{\omega_{i,t}} \Big|\Delta\phi^{\text{dewarp}}_{i,t} - \Delta\phi^{\text{dewarp}}_{i,t-1}\Big|, \forall t\in \mathcal{T}. 
	\end{equation} 
	Then, the fault-induced phase biases are obtained through temporal averaging: 
	$  \label{phase_1}
	\hat{\zeta}_i = \frac{1}{T_0} \sum_{t=1}^{T_0} \Delta\phi_{i,t}, \forall i \in \mathcal{D}
	$, leading to the following estimated HI coefficient vector
	\begin{align} \label{HI_estimated}
		\hat{c}_i =  \left\{\begin{array}{ll} e^{\jmath\hat{\zeta}_i},& \text{if} \;\;\; i\in \mathcal{D}, \\1,&\text{otherwise},\end{array}\right.
	\end{align}
	where $\hat{\mathbf{c}} = [\hat{c}_1,\dots,\hat{c}_N]^T$ denotes the estimated HI vector.
	
	\vspace{-5pt}
	\subsection{Phase~3: Target Localization} \vspace{-2pt} \label{Sec:Refined} 
	In this subsection, we perform fine-grained target localization based on calibrated phases.
	Specifically, we first decouple the joint angle and range estimates into separate estimates. 
	Then, based on calibrated phases, we construct a modified covariance matrix and design an efficient algorithm for sequential angle and range estimation.
	
	First, based on the calibrated phases in~\eqref{HI_estimated},  the calibrated received signal vector over all antennas is given by
	\begin{equation} \label{calibrated_signal}
		\hat{\mathbf{y}}_t = \hat{\mathbf{c}}^{\ast} \odot \mathbf{y}_t, \forall t \in \mathcal{T},
	\end{equation}
	where $(\cdot)^{\ast}$ denotes the conjugate operation.
	Then, to decouple the angle and range parameters, we first obtain an updated sample average of the covariance matrix $ {\mathbf{R}} \triangleq \frac{1}{T_0}\sum_{t=1}^{T_0}\hat{\mathbf{y}}_t\hat{\mathbf{y}}_t^{H} \in \mathbb{R}^{N\times N}$ based on calibrated signals in~\eqref{calibrated_signal}, for
	which the $(i,j)$-th entry is given by
	\begin{align} \label{covariance_R}
		\!\!\![{\mathbf{R}}]_{i,j} \!&=\!\sum_{k=1}^{K} g_k \exp \Big[\jmath \Big(\Delta_i\!-\!\Delta_j \! +\! \frac{2\pi}{\lambda}( (l_i-l_j) \sin(\theta_k)\!\!\! \Big.\Big. \nonumber\\
		&\Big.\Big.- (l_i-l_j)(l_i+l_j)  \frac{\cos^2(\theta_k)}{2 r_k} )\Big) \Big] + \sigma^2\delta_{i,j},
	\end{align}
	with $\delta_{i,j} = 1$ if $i=j$, and $\delta_{i,j} = 0$ otherwise.
	$\Delta_i \triangleq \zeta_i - \hat{\zeta}_i$ denotes the phase calibration error, which is very small and thus can be assumed negligible after phase calibration in Phase~2.
	Besides, due to the symmetry of the antenna positions, we have $l_i + l_j = 0$ if $ j = N+1-i$. 
	Thus, the $n$-th anti-diagonal entry of ${\mathbf{R}}$ is given by 
	\begin{align} \label{R_1} 
		[{\mathbf{R}}]_{n,N+1-n} 
		=\sum_{k=1}^{K} g_k e^{ \jmath \left( \frac{2\pi}{\lambda} 2 l_n \sin(\theta_k) \right)  } +  \sigma^2\delta_{n,N+1-n},
	\end{align} 
	which only involves the angle information, thus allowing us to separate the estimation of angle and range parameters.
	Then, we reorganize all the anti-diagonal entries in~\eqref{R_1} without the noise terms into a new vector ${\mathbf{r}} $ as follows
	\begin{equation}	
		[{\mathbf{r}}]_n =  \sum_{k=1}^{K} g_k e^ {\jmath \frac{2\pi}{\lambda} 2 l_n \sin(\theta_k)  },\forall n\in \mathcal{N}. 
	\end{equation}
	
	{ 
		To apply the MUSIC algorithm for high-resolution angle estimation, we need to construct a full-rank covariance matrix.
		However, since the vector $\mathbf{r}$ effectively represents a single snapshot of a virtual array, its direct covariance matrix would be of rand one. 
		To address this issue, similar to the method in~\cite{zhou2025mixednearfieldfarfieldtarget}, we divide ${\mathbf{r}}$ into $M(K\leq M\leq N-K)$ overlapping sub-vectors ${\mathbf{r}}_{m}, \forall m \in \mathcal{M}\triangleq [1,2,\dots,M]$ with each sub-vector containing $N-M+1$ entries}, which are given by
	\begin{align}
		{\mathbf{r}}_{m} &= \Big[ \sum_{k=1}^{K} g_k e^ {\jmath \frac{2\pi}{\lambda} 2 l_m \sin(\theta_k)},\dots, \sum_{k=1}^{K} g_k e^ {\jmath \frac{2\pi}{\lambda} 2 l_{N-M+{m}} \sin(\theta_k)  }\Big]^T\nonumber\\
		&= {\mathbf{B}}(\boldsymbol{\theta}) {\mathbf{w}}_{m}
		={\left[{\mathbf{b}} (\theta_1),\dots,{\mathbf{b}} (\theta_K) \right]} {\mathbf{w}}_{m}, \forall m \in \mathcal{M},
	\end{align}
	with ${\mathbf{b}}(\theta_k) = \big[1,\dots, e^{\jmath \frac{2\pi}{\lambda}2(N-M)d\sin(\theta_k) } \big]^T \in \mathbb{C}^{(N-M+1) \times 1}$ and 
	${\mathbf{w}}_{m} \!=\!\big[ g_1 e^ {\jmath \frac{2\pi}{\lambda} 2 l_m\! \sin(\theta_1)} ,\dots,  g_K e^ {\jmath \frac{2\pi}{\lambda} 2 l_m\! \sin(\theta_K\!)}\big]^T \!\in \mathbb{C}^{K\times 1}$.

	As such, with a total of $M$ sub-vectors, we can obtain the following covariance matrix $\mathbf{R}_{\mathbf{r}} \in \mathbb{C}^{(N-M+1)\times(N-M+1)}$ of ${\mathbf{r}}_m,\forall m \in \mathcal{M}$
	\vspace{-2pt}\begin{align} \label{R_tilde}
			\!\!{\mathbf{R}_{\mathbf{r}}}
			\!=\!\frac{1}{M} \mathbf{B}(\boldsymbol{\theta}) {\mathbf{W}} {\mathbf{B}}^H(\boldsymbol{\theta}) \!=\! \frac{1}{M} {\mathbf{B}}(\boldsymbol{\theta})\! \sum_{m=1}^{M}{\mathbf{w}}_{m}{\mathbf{w}}_{m}^H {\mathbf{B}}^H(\boldsymbol{\theta}).
	\end{align}
	Then we perform the EVD operation on the covariance matrix  ${\mathbf{R}_{\mathbf{r}}}$ to construct the noise subspace matrix ${\mathbf{U}}_{\mathrm{N},{\mathbf{r}}}\in\mathbb{C}^{(N-M+1)\times (N-M+1-K)}$, whose columns are the eigenvectors corresponding to smallest $(N-M+1-K)$ eigenvalues.
	The fine-grained estimated angles, denoted as $[\check{\theta}_1,\check{\theta}_2,\dots,\check{\theta}_K]$, can be obtained by finding $K$ largest peaks of the MUSIC spectrum around the coarse estimated angles $\{\hat{\theta}_k\},\forall k\in \mathcal{K}$:
	\begin{equation}
		\begin{aligned}\label{angle_spectrum_2}
			F(\theta) =   [{\mathbf{b}}^{H}(\theta){\mathbf{U}}_{\mathrm{N},{\mathbf{r}}}{\mathbf{U}}^H_{\mathrm{N},{\mathbf{r}}}{\mathbf{b}}(\theta)]^{-1}.
		\end{aligned}
	\end{equation}
	By substituting this $\{ \check{\theta}_k \}, \forall k \in \mathcal{K}$ into the covariance matrix in~\eqref{covariance_R}, we obtain
	${\mathbf{R}} = \mathbf{U}_{\mathrm{S}}\mathbf{\Sigma}_{\mathrm{S}}\mathbf{U}_{\mathrm{S}} + \mathbf{U}_{\mathrm{N}}\mathbf{\Sigma}_{\mathrm{N}}\mathbf{U}_{\mathrm{N}}$, where $\mathbf{U}_{\mathrm{S}}\in \mathbb{C}^{N\times K}$
	and  $\mathbf{U}_{\mathrm{N}}\in \mathbb{C}^{N\times (N-K)}$ denote the signal subspace and noise subspace, respectively, $\mathbf{\Sigma}_{\mathrm{S}}$ and $\mathbf{\Sigma}_{\mathrm{N}}$ denote diagonal matrices containing $K$ largest and $N-K$ smallest eigenvalues, respectively.
	Then, the ranges of targets can be obtained by finding peaks from range spectrum as follows
	\begin{equation}\vspace{-2pt}
		\begin{aligned}\label{distance_spectrum}
			\check{r}_k =  \underset{r}{\operatorname*{argmax}} \big[\mathbf{a}^{H}(\check{\theta}_k,r){\mathbf{U}}_{\mathrm{N}}{\mathbf{U}}^H_{\mathrm{N}}\mathbf{a}(\check{\theta}_k,r)\big]^{-1}, \forall k \in \mathcal{K},
		\end{aligned}\vspace{-2pt}
	\end{equation}
	hence yielding fine-grained targets' positions $(\check{\theta}_k,\check{r}_k), \forall k\in \mathcal{K}$.
	
	\vspace{-8pt}
	\subsection{Extension and Discussions}
	\vspace{-4pt}
	\subsubsection{\underline{\textbf{Passive localization}}}\label{activelocaliza}
	The proposed algorithm can be extended to the passive localization scenario where the BS sends broadcast signals and localizes targets based on echo signals~\cite{86917}.
	Let $\mathbf{s}_t$ denote the transmitted signal by the BS at time $t$.
	Then, the echo signal (reflected by targets) received by the BS at time index $t$ can be modeled as
	\begin{align}\label{echo_signal}
		\mathbf{y}_t &=  \sum_{k=1}^{K}  \alpha_k (\beta_k \mathbf{c} \odot \mathbf{a}(\theta_k,r_k) ) \left(\beta_k \mathbf{c} \odot \mathbf{a}(\theta_k,r_k)\right)^{H}\mathbf{s}_{t} + \mathbf{n}_t \nonumber \\
		&= \diag(\mathbf{c})\mathbf{A}(\boldsymbol{\theta},\mathbf{r}){\mathbf{x}}_t + \mathbf{n}_t, \forall t\in\mathcal{T},
	\end{align}
	where $\mathbf{x}_t = \left[ \alpha_1|\beta_1|^2{x}_{1,t},\alpha_2|\beta_2|^2{x}_{2,t},\dots, \alpha_K|\beta_K|^2{x}_{K,t} \right]^T$ denotes the equivalent transmitted signal of the $K$ targets with $\alpha_k$ denoting the gain coefficient of the radar cross section (RCS) and 	${x}_{k,t} = \left(\mathbf{c} \odot \mathbf{a}(\theta_k,r_k)\right)^{H}\mathbf{s}_{t} $ denoting the equivalent transmitted signal at the $k$-th target. \\
	\indent It is worth noting that the received signal at the BS exhibits a form similar to that in~\eqref{received_signal}. 
	The impact of passive localization on our proposed method mainly lies in the covariance matrix of the equivalent transmitted signal ${\mathbf{x}}_t$.
	Therefore, similar to Section~\ref{Sec4}, the fault detection in Phase~1 and fault phase calibration in Phase~2 can be applied directly without significant modification.
	In Phase~3, with calibrated received signal vector $\hat{\mathbf{y}}_t$ in~\eqref{calibrated_signal}, the covariance matrix of the signal $\hat{\mathbf{y}}_t$ is given by $
	{\mathbf{R}}_{\hat{y}} =  \mathbb{E}_t\left[ \hat{\mathbf{y}}_t \hat{\mathbf{y}}^H_t\right] = \mathbf{A}\mathbf{R}_{{\mathbf{x}}}\mathbf{A}^H + \sigma^2\mathbf{I}$,
	where $\mathbf{R}_{{\mathbf{x}}} = \mathbb{E}_t\left[{\mathbf{x}}_t{\mathbf{x}}^H_t\right]$ denotes the covariance matrix of the equivalent transmitted signal vector, for which the $(i,j)$-th entry can be expressed as 
	\begin{equation} \label{covariance_echo_trans}
		\begin{aligned}
			\!\!\left[\mathbf{R}_{{\mathbf{x}}}\right]_{i,j} = \alpha_i \alpha^{\ast}_j |\beta_i|^2 |\beta_j|^2 \mathbf{a}^H(\theta_i,r_i)\mathbf{a}(\theta_j,r_j), \forall i,j\in \mathcal{N}.
		\end{aligned} 
	\end{equation}
	It is worth noting that the correlation between $\mathbf{a}(\theta_i,r_i)$ and $\mathbf{a}(\theta_j,r_j)$ is negligible if the targets are widely separated~\cite{9693928}.
	In this case, $\mathbf{R}_{{\mathbf{x}}}$ is approximately a diagonal matrix, which can be expressed as $\mathbf{R}_{{\mathbf{x}}} \approx  \diag\left\{|\alpha_1|^2 |\beta_1|^4 N,|\alpha_2|^2 |\beta_2|^4 N, \dots, |\alpha_K|^2|\beta_K|^4 N\right\}$.
	Hence, the original covariance matrix of the echo signals in~\eqref{covariance_echo_trans} has the same form as the covariance matrix of signals in~\eqref{received_signal}, rendering the proposed method directly applicable.
	
	On the other hand, if the targets are close to each other, the correlation between $\mathbf{a}(\theta_i,r_i)$ and $\mathbf{a}(\theta_j,r_j)$ cannot be neglected, resulting in a rank-deficient covariance matrix $\mathbf{R}_{{\mathbf{x}}}$.
	Consequently, the dimensionality of the signal subspace is underestimated, making the classic MUSIC algorithm unable to resolve individual targets.
	One possible solution is using the spatial smoothing technique to achieve a full rank covariance matrix~\cite{11146870}.
	Alternatively, another approach is by using the Toeplitz approximation method, which restores the full-rank property by enforcing the inherent structure of the covariance matrix~\cite{6218750}. 
	This method enforces a Hermitian-Toeplitz structure by averaging the elements along each diagonal of the sample covariance matrix, thereby decorrelating coherent signals and restoring the matrix rank.

	\subsubsection{\underline{\textbf{Ambiguous angle}}} 
	Conventional near-field localization methods typically require the inter-antenna spacing to be no larger than $\lambda/4$ for avoiding phase ambiguity~\cite{5200332}. 
	However, when the inter-antenna spacing is greater than or equal to $\lambda/4$, the MUSIC algorithm, after decoupling angle and range information, exhibits periodic spectral peaks~\cite{11146870}.
	This phenomenon obstructs the identification of true angles, hence hindering the application in practice. 
	{Fortunately, the proposed three-phase algorithm helps avoid this phase ambiguity issue by leveraging distinct characteristics of signals in each localization phase.}
	Specifically, in the first phase of fault detection, the targets are situated in the far-field region of the subarrays, which obviates the need to decouple angle and range information for initial estimation. 
	In Phase~3, the spectral peak search for both angle and range is confined to the vicinity of the coarse estimates obtained in Phase~1. By strategically constraining the search space, the proposed algorithm avoids exploring other spectrum regions where periodic and ambiguous peaks might exist.

	

	\subsubsection{\underline{\textbf{Computational complexity}}}
	The computational complexity of the proposed algorithm is summarized as follows.
	For Phase~1, the computational complexity of fault detection is dominated by the gradient update in~\eqref{gradient_f} with an order of $\mathcal{O}\left( T_{0}NK\right) $. The  computational complexity of coarse localization is dominated by covariance matrices construction in~\eqref{eq:covariance_mat12}, EVD in~\eqref{EVD_R_q}, MUSIC spectrum search in~\eqref{angle_spectrum}, and channel gain estimation in~\eqref{channel_gain_solution}, with their scaling orders given by $\mathcal{O}(T_{0}N_{\mathrm{sub}}^2)$, $\mathcal{O}(N_{\mathrm{sub}}^3)$, $\mathcal{O}(I_{\mathrm{sub}}N_{\mathrm{sub}}^2)$, and $\mathcal{O}(NT_0K^2)$ respectively, where $I_{\mathrm{sub}}$ denotes the number of search grids.
	As such, considering that $N_{\text{sub}} = 0.6N^{\frac{3}{4}}$, the computational complexity of Phase~1 is in the order of $\mathcal{O}(NN_{\mathrm{iter}}(T_{0}K + T_{0}N^{\frac{3}{4}} + N^{\frac{3}{2}} +  N^{\frac{3}{4}}I_{\mathrm{sub}}+T_0K^2))$,
	where $N_{\mathrm{iter}}$ denotes the maximum number of iterations to achieve convergence.
	For Phase~2, the average number of faulty antennas is $Np_{\text{fault}}$, thus the computational complexity order of phase calibration is $\mathcal{O}(T_{0}Np_{\text{fault}})$.
	For Phase~3, the computational complexity of fine-grained localization is dominated by the covariance matrix construction in~\eqref{R_tilde}, the EVD operation, as well as the angle and range parameter estimation, which is in the order of $\mathcal{O}(N^2T_{0}+ (N-M+1)^3+I_{\theta}K(N-M+1)^2+I_{r}K(N-M+1)^2)$~\cite{10778649}, where $I_{\theta}$ and $I_{r}$ represent the numbers of search grids in the angle and range domains, respectively.
	Based on the above, the overall computational complexity of the proposed algorithm is $\mathcal{O}(NN_{\mathrm{iter}}(T_{0}K + T_{0}N^{\frac{3}{4}} + N^{\frac{3}{2}} +  N^{\frac{3}{4}}I_{\mathrm{sub}}) +T_{0}N(p_{\text{fault}} + N) + (I_{\theta}K+ I_{r}K+N-M+1)(N-M+1)^2)$. 
	In the near-field localization scenarios with XL-array, the overall asymptotic complexity of the proposed algorithm is dominated by the fine-grained localization in Phase 3 in the order of $O(N^{3})$, mainly due to the covariance matrix construction and the EVD operation on the XL-array.
	
	\vspace{-5pt}
	\section{Cram\'{e}r-Rao Bound Analysis} \label{Sec5}
	In this section, we characterize the near-field localization performance under two distinct conditions, namely, without prior HI information and with perfect HI information.

	\vspace{-10pt}
	\subsection{Scenarios without Prior HI Information} \vspace{-3pt}
	When HI information is unknown, directly using the localization algorithm will cause model misspecification, which means that the assumed model used to design the localization algorithm is inconsistent with the true model with HIs.
	In this case, localization algorithms based on assumed models typically do not converge to the true parameters even when the number of snapshots approaches infinity, but rather provide \textit{pseudo-true} parameters~\cite{7401095}.
	As such, we present the MCRB analysis for scenarios without HI information to characterize the localization performance degradation caused by model mismatch, i.e., the discrepancy between the assumed fault-free model and the true model with faulty antennas.

	We denote $\overline{\boldsymbol{\eta}} \triangleq [{\Re}( \overline{\boldsymbol{\beta}})^T,{\Im}(\overline{\boldsymbol{\beta}})^T,\overline{\boldsymbol{\theta}}^T,\overline{\mathbf{r}}^T]^T \in \mathbb{R}^{4K\times 1}$ as
	the true values of the unknown parameters $\boldsymbol{\eta} \triangleq [{\Re}( {\boldsymbol{\beta}})^T,{\Im}({\boldsymbol{\beta}})^T,\boldsymbol{\theta}^T,\mathbf{r}^T]^T \in \mathbb{R}^{4K\times 1}$. Then the received signal of the true model is given by
	\begin{equation} \label{true_model}
		\begin{aligned}
			\mathbf{y}_t =  \diag(\mathbf{c})\mathbf{A}(\overline{\boldsymbol{\theta}},\overline{\mathbf{r}})\diag(\overline{\boldsymbol{\beta}})\mathbf{s}_t + \mathbf{n}_t, \forall t \in \mathcal{T}.
		\end{aligned}
	\end{equation}
	By denoting $\boldsymbol{\mu}_t = \diag(\mathbf{c})\mathbf{A}(\overline{\boldsymbol{\theta}},\overline{\mathbf{r}}) \diag(\overline{\boldsymbol{\beta}}) \mathbf{s}_t, \forall t \in \mathcal{T}$, for any given HI coefficient vector $\mathbf{c}$, the probability density function (PDF) of the observation in~\eqref{true_model} can be expressed as
	\begin{equation} \label{true_pdf}
		\begin{aligned}
			p(\mathbf{y}) = \frac{1}{\sqrt{2\pi\sigma^2}^{NT_0}}\exp\left\{-\frac{\|\mathbf{y} - \boldsymbol{\mu}\|^2}{2\sigma^2}\right\},
		\end{aligned}
	\end{equation}
	where $\bmu\triangleq [\boldsymbol{\mu}_1^T, \dots , \boldsymbol{\mu}_{T_0}^T]^{T}\in \mathbb{C}^{NT_0\times 1}$ and $\mathbf{y} \triangleq [\mathbf{y}_1^{T},\dots, \mathbf{y}_{T_0}^{T}]^{T}\in \mathbb{C}^{NT_{0}\times 1}$.
	
	Next, for the case without prior HI information, we further assume that there are no faulty antennas, i.e., $\mathbf{c} = 1$. As such, the received signal is given by
	$
	\mathbf{y}_t =  \mathbf{A}({\boldsymbol{\theta}},{\mathbf{r}})\diag({\boldsymbol{\beta}})\mathbf{x}_t + \mathbf{n}_t, \forall t \in \mathcal{T}
	$. The observation PDF under the assumed model can be expressed as \vspace{-2pt}
	\begin{equation}\label{mis_pdf}
		\begin{aligned}
			\tilde{p}(\mathbf{y}|\boldsymbol{\eta}) = \frac{1}{\sqrt{2\pi\sigma^2}^{NT_0}}\exp\left\{-\frac{\|\mathbf{y} - \tilde{\boldsymbol{\mu}}(\boldsymbol{\eta})\|^2}{2\sigma^2}\right\},
		\end{aligned}
	\end{equation}
	where $\tilde{\boldsymbol{\mu}}(\boldsymbol{\eta}) \triangleq [\tilde{\boldsymbol{\mu}}_1^T(\boldsymbol{\eta}), \dots , \tilde{\boldsymbol{\mu}}_{T_0}^T(\boldsymbol{\eta})]^{T} \in \mathbb{C}^{NT_0\times 1}$ with $ \tilde{\boldsymbol{\mu}}_t(\boldsymbol{\eta})\triangleq \mathbf{A}({\boldsymbol{\theta}},{\mathbf{r}})\diag({\boldsymbol{\beta}})\mathbf{s}_t,\forall t\in \mathcal{T}$.

	\begin{definition}
		\emph{Given the true PDF in~\eqref{true_pdf} and the assumed PDF in~\eqref{mis_pdf}, we define the \textit{pseudo-true} parameter that minimizes the Kullback-Leibler (KL) divergence between the assumed PDF in~\eqref{mis_pdf} and the true PDF in~\eqref{true_pdf}, which is given by~\cite{10384355}
			\begin{equation}\label{eq:eta0}
				\dot{\boldsymbol{\eta}} = \arg\min_{\bet\in\mathbb{R}^{4K\times 1}} \text{KL}\left(p(\mathbf{y})\Vert \tilde{p}(\mathbf{y}|\boldsymbol{\eta})\right),
			\end{equation}
			where the KL divergence $ \text{KL}\left(p(\mathbf{y})\Vert \tilde{p}(\mathbf{y}|\boldsymbol{\eta})\right) = \int p(\mathbf{y}) \log\left(\frac{p(\mathbf{y})}{\tilde{p}(\mathbf{y}|\boldsymbol{\eta})}\right)d\mathbf{y} $ characterizes how much the assumed  PDF $\tilde{p}(\mathbf{y}|\boldsymbol{\eta})$ is different from the true PDF $p(\mathbf{y})$.
		}
	\end{definition}
	
	\begin{definition}
		\emph{MCRB provides a lower bound for any misspecified-unbiased estimator $\hat{\boldsymbol{\eta}}(\mathbf{y})$ of the pseudo-true parameter $\dot{\boldsymbol{\eta}}$, i.e., $				\mathbb{E}_{\mathbf{y}}\{\left(\hat{\boldsymbol{\eta}}(\mathbf{y})-\dot{\boldsymbol{\eta}}\right)\left(\hat{\boldsymbol{\eta}}(\mathbf{y})-\dot{\boldsymbol{\eta}}\right)^T\} \succeq \mathbf{MCRB}(\dot{\boldsymbol{\eta}})$~\cite{8103120}. 
			Mathematically, MCRB is given by the general Fisher information matrix (FIM) as follows~\cite{10384355}
			\begin{align} \label{eq_mcrb_def}
				\mathbf{MCRB}(\dot{\boldsymbol{\eta}}) \triangleq {\mathbf{C}^{-1}(\dot{\boldsymbol{\eta}})}\tilde{\mathbf{J}}(\dot{\boldsymbol{\eta}}){\mathbf{C}^{-1}(\dot{\boldsymbol{\eta}})},
			\end{align}
			where $
			{\mathbf{C}}(\dot{\boldsymbol{\eta}}) \triangleq \mathbb{E}_{\mathbf{y}}\left\{\frac{\partial^2}{\partial \boldsymbol{\eta} \partial \boldsymbol{\eta}^T} \ln  \tilde{p}(\mathbf{y}|\boldsymbol{\eta})  \big|_{\boldsymbol{\eta} = \dot{\boldsymbol{\eta}}} \right\}$,
			and 
			$
			\tilde{\mathbf{J}}(\dot{\boldsymbol{\eta}}) \triangleq \mathbb{E}_{\mathbf{y}}\left\{\left(\frac{\partial}{\partial \boldsymbol{\eta}} \ln  \tilde{p}(\mathbf{y}|\boldsymbol{\eta}) \right)\left(\frac{\partial}{\partial \boldsymbol{\eta}} \ln  \tilde{p}(\mathbf{y}|\boldsymbol{\eta})\right)^{T} \Big|_{\boldsymbol{\eta} = \dot{\boldsymbol{\eta}}} \right\}
			$, with $\mathbb{E}_{\mathbf{y}}\{ \cdot \}$ denoting the expectation operator under $p(\mathbf{y})$ in~\eqref{true_pdf}.}
	\end{definition}
	
	\begin{lemma}\label{lemma2}\emph{
			Given any misspecified estimator $\hat{\boldsymbol{\eta}}(\mathbf{y})$, which is unbiased of the pseudo-true parameter $\dot{\boldsymbol{\eta}}$, i.e., $\mathbb{E}_{\mathbf{y}}\{\hat{\boldsymbol{\eta}}(\mathbf{y})\} = \dot{\boldsymbol{\eta}}\neq \overline{\boldsymbol{\eta}}$, the mean-squared error (MSE) under scenarios without HI information can be lower-bounded as
			\begin{align}  \label{LB}
				&\mathbb{E}_{\mathbf{y}}\{\left(\hat{\boldsymbol{\eta}}(\mathbf{y})-\overline{\boldsymbol{\eta}}\right)\left(\hat{\boldsymbol{\eta}}(\mathbf{y})-\overline{\boldsymbol{\eta}}\right)^T\} \nonumber\\
				& \succeq \mathbf{Bias}(\dot{\boldsymbol{\eta}})
				+\mathbf{MCRB}(\dot{\boldsymbol{\eta}}) 
				\triangleq \mathbf{LB}(\dot{\boldsymbol{\eta}}).
			\end{align}
			Herein, $\mathbf{LB}(\dot{\boldsymbol{\eta}})$ denotes the LB of MSE, $\mathbf{Bias}(\dot{\boldsymbol{\eta}})$ denotes the bias between true values $ \overline{\boldsymbol{\eta}}$ and pseudo-true values $\dot{\boldsymbol{\eta}}$, which is given by\vspace{-6pt}
			\begin{align} \label{Bias}
				\mathbf{Bias}(\dot{\boldsymbol{\eta}}) = (\dot{\boldsymbol{\eta}} - \overline{\boldsymbol{\eta}})(\dot{\boldsymbol{\eta}} - \overline{\boldsymbol{\eta}})^T.
			\end{align}
		}
	\end{lemma}\vspace{-2pt}
	
	{\textbf{Lemma \ref{lemma2}} shows that there are two key components of estimation errors under HIs: the variance bound (i.e. $\mathbf{MCRB}(\dot{\boldsymbol{\eta}}) $), and the squared bias term (i.e., $\mathbf{Bias}(\dot{\boldsymbol{\eta}})$). Although the MCRB term may decrease with improved SNR, the bias term—arising from the inherent deviation between the assumed ideal array and the true faulty array—remains constant. 
	This indicates that increasing transmit power alone is insufficient to eliminate localization errors caused by HIs. 
	Thus, explicit calibration and HI-aware processing techniques are essential, which will be further demonstrated by numerical results.}

	To obtain the LB in~\eqref{LB}, the entries of $\mathbf{C}(\dot{\boldsymbol{\eta}})$ and $\tilde{\mathbf{J}}(\dot{\boldsymbol{\eta}})$ in~\eqref{eq_mcrb_def} are obtained in~\eqref{C} and~\eqref{J}, respectively, shown on top of the next page, where $\frac{\partial\tilde{\boldsymbol{\mu}}(\boldsymbol{\eta})}{\partial\eta_i} = \Big[\big(\frac{\partial\tilde{\boldsymbol{\mu}}_1(\boldsymbol{\eta})}{\partial\eta_i}\big)^T,\dots,\big(\frac{\partial\tilde{\boldsymbol{\mu}}_{T_0}(\boldsymbol{\eta})}{\partial\eta_i}\big)^T\Big]^T$,  $ \frac{\partial^2\tilde{\boldsymbol{\mu}}\left(\boldsymbol{\eta}\right) }{\partial\eta_i\partial \eta_j} = \Big[\big(\frac{\partial^2\tilde{\boldsymbol{\mu}}_1\left(\boldsymbol{\eta}\right) }{\partial\eta_i\partial \eta_j}\big)^T,\dots,\big(\frac{\partial^2\tilde{\boldsymbol{\mu}}_{T_0}\left(\boldsymbol{\eta}\right) }{\partial\eta_i\partial \eta_j}\big)^T\Big]^T$, and $\boldsymbol{\epsilon}(\boldsymbol{\eta}) = {\boldsymbol{\mu}}- \tilde{\boldsymbol{\mu}}(\boldsymbol{\eta})$.
	\setcounter{equation}{\value{equation}} 
	\begin{figure*}[h]
		\begin{equation}
			\begin{aligned} \label{C}
				\left[\mathbf{C}(\dot{\boldsymbol{\eta}})\right]_{i,j} = \mathbb{E}_{\mathbf{y}}\Big\{\frac{\partial^2}{\partial{\eta}_i\partial{\eta}_j} \ln  \tilde{p}(\mathbf{y}|\boldsymbol{\eta}) \Big|_{\boldsymbol{\eta} = \dot{\boldsymbol{\eta}}}\Big\}
				=  \frac{1}{\sigma^2} \Re\Big\{ \left(\frac{\partial^2\tilde{\boldsymbol{\mu}}\left(\boldsymbol{\eta}\right) }{\partial\eta_i\partial \eta_j}\right)^T\!\!\!\boldsymbol{\epsilon}(\boldsymbol{\eta})\!-\! \left(\frac{\partial\tilde{\boldsymbol{\mu}}(\boldsymbol{\eta})}{\partial\eta_i}\right)^T\!\!\frac{\partial\tilde{\boldsymbol{\mu}}(\boldsymbol{\eta})}{\partial\eta_j}\Big\}\Big|_{\boldsymbol{\eta} = \dot{\boldsymbol{\eta}}},
			\end{aligned}
		\end{equation}\vspace{-5pt} 
		\begin{align} \label{J}
			\left[\mathbf{J}(\dot{\boldsymbol{\eta}})\right]_{i,j} =& \mathbb{E}_{\mathbf{y}}\Big\{\frac{\partial}{\partial{\eta}_i} \ln \tilde{p}(\mathbf{y}|\boldsymbol{\eta}) \frac{\partial}{ \partial{\eta_j}}\ln \tilde{p}(\mathbf{y}|\boldsymbol{\eta}) \Big|_{\boldsymbol{\eta} = \dot{\boldsymbol{\eta}}}\Big\}\nonumber\\
			=& \frac{1}{\sigma^4}\Re\Big\{ \boldsymbol{\epsilon}^T(\boldsymbol{\eta}) \frac{\partial\tilde{\boldsymbol{\mu}}(\boldsymbol{\eta})}{\partial\eta_i}  \Big\} \Re \Big\{\boldsymbol{\epsilon}^T(\boldsymbol{\eta}) \frac{\partial\tilde{\boldsymbol{\mu}}(\boldsymbol{\eta})}{\partial\eta_j} \Big\} 
			+ \frac{1}{\sigma^2}\Re\Big\{ \left(\frac{\partial\tilde{\boldsymbol{\mu}}(\boldsymbol{\eta})}{\partial\eta_i}\right)^T\frac{\partial\tilde{\boldsymbol{\mu}}(\boldsymbol{\eta})}{\partial\eta_j} \Big\} \Big|_{\boldsymbol{\eta} = \dot{\boldsymbol{\eta}}}.
		\end{align}\vspace{-16pt} 
		\hrulefill 
	\end{figure*}  
	We give the first and second derivatives of  $\tilde{\boldsymbol{\mu}}_t(\boldsymbol{\eta})$ w.r.t. $\eta_i, \forall i=1,\dots,4K$ in Appendix~\ref{App2}.
	As such, we can obtain $\mathbf{C}(\dot{\boldsymbol{\eta}})$ in~\eqref{C} and $\mathbf{J}(\dot{\boldsymbol{\eta}})$ in~\eqref{J} based on these derivatives, and then obtain the LB.

	\vspace{-6pt}
	\subsection{Scenarios with Perfect HI Information} \vspace{-2pt}
	In this section, we present standard CRB analysis, which characterizes the theoretical performance when perfect knowledge of faulty antennas is available after phase calibration~\cite{10388218}.
	
	\begin{lemma}\emph{
			The CRB for parameters $\boldsymbol{\eta}$ is determined by the inversion of the FIM~\cite{11146870}, which is given by
			\begin{align}\label{FIM_CRB}
				\!\!\!{\mathbf{J}}({\boldsymbol{\eta}}) &=\frac{1}{\sigma^2} \Re \bigg\{\Big(\frac{\partial \breve{\boldsymbol{\mu}}(\boldsymbol{\eta})}{\partial \boldsymbol{\eta}}  \Big)^H \Big(\frac{\partial \breve{\boldsymbol{\mu}}(\boldsymbol{\eta})}{\partial \boldsymbol{\eta}} \Big) \bigg\},
			\end{align}
			where $\breve{\boldsymbol{\mu}}(\boldsymbol{\eta}) \triangleq [\breve{\boldsymbol{\mu}}_1^T(\boldsymbol{\eta}), \dots , \breve{\boldsymbol{\mu}}_{T_0}^T(\boldsymbol{\eta})]^{T} \in \mathbb{C}^{NT_0\times 1}$ with $ \breve{\boldsymbol{\mu}}_t(\boldsymbol{\eta})\triangleq\diag(\mathbf{c}) \mathbf{A}({\boldsymbol{\theta}},{\mathbf{r}})\diag({\boldsymbol{\beta}})\mathbf{s}_t,\forall t\in \mathcal{T}$.
			The CRB on parameter $\boldsymbol{\eta}$ is given by
			\begin{equation}\label{mutiCRB}
				\mathbf{CRB}({\boldsymbol{\eta}})= \tr\Big\{\big[ {\mathbf{J}}^{-1} ({\boldsymbol{\eta}} )\big]_{} \Big\}.
			\end{equation}
			 With the derivatives obtained in Appendix~\ref{App2},~\eqref{mutiCRB} provides a matrix closed-form expression for the CRBs of multi-target localization.
		}
	\end{lemma}

	\begin{remark}\label{prop_1}\emph{(Effect of $p_{\text{fault}}$). 
			Note that it is difficult, if not possible, to obtain a closed-form expression to characterize the relationship between the fault probability $p_{\text{fault}}$ and MCRB, since the positions of faulty antennas are random and the effects are complicated. 
			As such, we provide qualitative discussion for the effect of $p_{\text{fault}}$ on MCRB.
			Specifically, when the fault probability decreases, the observation PDF under the true model in~\eqref{true_pdf} becomes more similar to that under the assumed model in~\eqref{mis_pdf}.
			This causes  $\dot{\boldsymbol{\eta}} $ to approach ${\boldsymbol{\eta}} $, thereby leading to decreasing bias term in~\eqref{Bias}.
			In the extreme case with $p_{\text{fault}} = 0$, i.e., there are no faulty antennas, the mismatch between the true model and the assumed model disappears, which means $\tilde{p}(\mathbf{y}|\boldsymbol{\eta}) =  p(\mathbf{y})$.
			In this case, the pseudo-true parameter $\dot{\boldsymbol{\eta}}$ in~\eqref{eq:eta0} equals the true parameter $\overline{\boldsymbol{\eta}}$, i.e.,		  $\dot{\boldsymbol{\eta}} = \overline{\boldsymbol{\eta}}$, which makes the bias term in~\eqref{Bias} be zero.
			Besides, the matrices $\mathbf{C}(\overline{\boldsymbol{\eta}})$ and $\tilde{\mathbf{J}}(\overline{\boldsymbol{\eta}})$ in~\eqref{eq_mcrb_def} can be re-expressed as 
			$
			\mathbf{C}(\overline{\boldsymbol{\eta}}) = \mathbb{E}_{\mathbf{y}}\big\{\frac{\partial^2}{\partial \boldsymbol{\eta} \partial \boldsymbol{\eta}^T} \ln  {p}(\mathbf{y}|\boldsymbol{\eta})  \big|_{\boldsymbol{\eta} = \overline{\boldsymbol{\eta}}} \big\}
			$ and 
			$
			\tilde{\mathbf{J}}(\overline{\boldsymbol{\eta}}) = \mathbb{E}_{\mathbf{y}}\big\{\big(\frac{\partial}{\partial \boldsymbol{\eta}} \ln  {p}(\mathbf{y}|\boldsymbol{\eta}) \big)\big(\frac{\partial}{\partial \boldsymbol{\eta}} \ln  {p}(\mathbf{y}|\boldsymbol{\eta})\big)^{T} \big|_{\boldsymbol{\eta} = \overline{\boldsymbol{\eta}}} \big\}
			$. 
			Thus, we have $\mathbf{C}(\overline{\boldsymbol{\eta}}) = {\mathbf{J}}(\overline{\boldsymbol{\eta}})$ and $ \tilde{\mathbf{J}}(\overline{\boldsymbol{\eta}})= {\mathbf{J}}(\overline{\boldsymbol{\eta}})$.
			By substituting them into~\eqref{eq_mcrb_def}, MCRB in~\eqref{eq_mcrb_def} can be re-expressed as $	\mathbf{MCRB}(\overline{\boldsymbol{\eta}}) = {\mathbf{C}^{-1}(\overline{\boldsymbol{\eta}})}\tilde{\mathbf{J}}(\overline{\boldsymbol{\eta}}){\mathbf{C}^{-1}(\overline{\boldsymbol{\eta}})} = {{\mathbf{J}}^{-1}(\overline{\boldsymbol{\eta}})}{\mathbf{J}}(\overline{\boldsymbol{\eta}}){{\mathbf{J}}^{-1}(\overline{\boldsymbol{\eta}})}  = {\mathbf{J}}^{-1}(\overline{\boldsymbol{\eta}}) =  \mathbf{CRB}(\overline{\boldsymbol{\eta}})$.
			As such, we have $\mathbf{Bias} = \mathbf{0}$, and $\mathbf{LB} = \mathbf{MCRB} = \mathbf{CRB}$.
		}
	\end{remark}
	

	\vspace{-10pt}
	\section{Numerical Results}\label{Sec:SR} 
	In this section, we present numerical results to demonstrate the effectiveness of our proposed near-field target localization scheme under HIs.

	\vspace{-10pt}
	\subsection{System Setup and Benchmark Schemes}
	The system setup is as follows unless otherwise specified.
	We consider an XL-array with $N = 256$ antennas and operating at a frequency of $f=30$ GHz. There are three targets located in the near-field of the XL-array, with their positions given by $(\frac{\pi}{12} \text{ rad},10 \text{ m})$, $(-\frac{\pi}{12} \text{ rad},18 \text{ m})$, $(-\frac{\pi}{6} \text{ rad},8 \text{ m})$, respectively.
	The number of snapshots is set as $T_0=100$, and the numerical results are averaged over $\Upsilon = 500$ Monte Carlo simulations. 
	{The update step size $\nu$ and regularization parameter $\rho$ in Phase 1 are set as 0.0001 and 0.01, respectively.}
	For each simulation, we denote by $(\check{\theta}^{(i)}_{k}, \check{r}^{(i)}_k)$  the estimation of $({\theta_{k}}, {r_k})$ in the polar coordinate system and $\check{\boldsymbol{\xi}}^{(i)}_k = [\check{r}^{(i)}_k\sin(\check{\theta}^{(i)}_{k}),\check{r}^{(i)}_k\cos(\check{\theta}^{(i)}_{k})]^T$ the estimation of $ \boldsymbol{\xi}_k= [r_k \sin(\theta_k),r_k\cos(\theta_k)]^T$ in the Cartesian coordinate system.
	The received SNR is defined as $\mathsf{SNR} = \frac{\mathbb{E}_t\{\|\mathbf{y}_t\|^2\}}{N\sigma^2}$. 
	Three performance metrics are considered, 1) the localization root MSE (RMSE) $\boldsymbol{\xi}_{\mathrm{RMSE}} = \sqrt{ \frac{1}{\Upsilon K}\sum_{i}^{\Upsilon} \sum_{k}^{K} \| \check{\boldsymbol{\xi}}^{(i)}_k - \boldsymbol{\xi}_k \|^2 }
	$; 2) angle estimation RMSE $
	{r}_{\mathrm{RMSE}} = \sqrt{ \frac{1}{\Upsilon K}\sum_{i}^{\Upsilon} \sum_{k}^{K} \| \check{r}^{(i)}_k - {r}_k \|^2 }$; and 3) range estimation RMSE $
	{\theta}_{\mathrm{RMSE}} = \sqrt{ \frac{1}{\Upsilon K}\sum_{i}^{\Upsilon} \sum_{k}^{K} \| \check{\theta}^{(i)}_{k} - {\theta}_{k} \|^2 }
	$.
	
	For performance comparison, we consider the following benchmark schemes and CRBs:
	\begin{itemize}
		\item \textbf{(BCD-based scheme)} The BCD-based fault detection and localization scheme alternately optimizes HI indicator, channel gains, angles and ranges of targets as in~\cite{10384355}. 
		
		\item \textbf{(Coarse localization scheme)} The coarse localization scheme is obtained by the method in Section~\ref{Sec:Coarse}.
		\item \textbf{(Root-CRB)} Root-CRB (R-CRB) is used to serve as the estimation performance lower bound with perfect faulty-antenna information in~\eqref{mutiCRB}, which is given by $\text{R-CRB}_{\text{angle}} = \sqrt{\sum_{k=2K+1}^{3K} [\mathbf{CRB}]_{k,k} }$ and $\text{R-CRB}_{\text{range}} = \sqrt{\sum_{k=3K+1}^{4K} [\mathbf{CRB}]_{k,k} }$.
		\item \textbf{(Root-LB)} Root-LB (R-LB) is used to serve as the estimation performance lower bound without prior HI information in~\eqref{LB}, which is given by $\text{R-LB}_{\text{angle}} = \sqrt{\sum_{k=2K+1}^{3K} [\mathbf{LB}]_{k,k} }$ and $\text{R-LB}_{\text{range}} = \sqrt{\sum_{k=3K+1}^{4K} [\mathbf{LB}]_{k,k} }$. 
		\item \textbf{(Root-MCRB)} Root-MCRB (R-MCRB) is
		used to quantify the estimation performance under the assumed fault-free model in~\eqref{eq_mcrb_def}, which is given by 
		$\text{R-MCRB}_{\text{angle}} = \sqrt{\sum_{k=2K+1}^{3K} [\mathbf{MCRB}]_{k,k} }$ and $\text{R-MCRB}_{\text{range}} = \sqrt{\sum_{k=3K+1}^{4K} [\mathbf{MCRB}]_{k,k} }$.
		\item \textbf{(Root-bias)} Root-bias (R-bias) is used to quantify the fixed discrepancy between the expected value of estimated pseudo-true values and true values in~\eqref{Bias}, which is given by $\text{R-bias}_{\text{angle}} = \sqrt{\sum_{k=2K+1}^{3K} [\mathbf{Bias}]_{k,k} }$ and $\text{R-bias}_{\text{range}} = \sqrt{\sum_{k=3K+1}^{4K} [\mathbf{Bais}]_{k,k} }$.
	\end{itemize}

	
	

	\begin{figure*}[t]	
		\vspace{-10pt}
		\centering
		\subfigure[Angle estimation RMSE versus SNR.]{	\label{fig:anglermsevssnr}
			\includegraphics[width=0.35\linewidth]{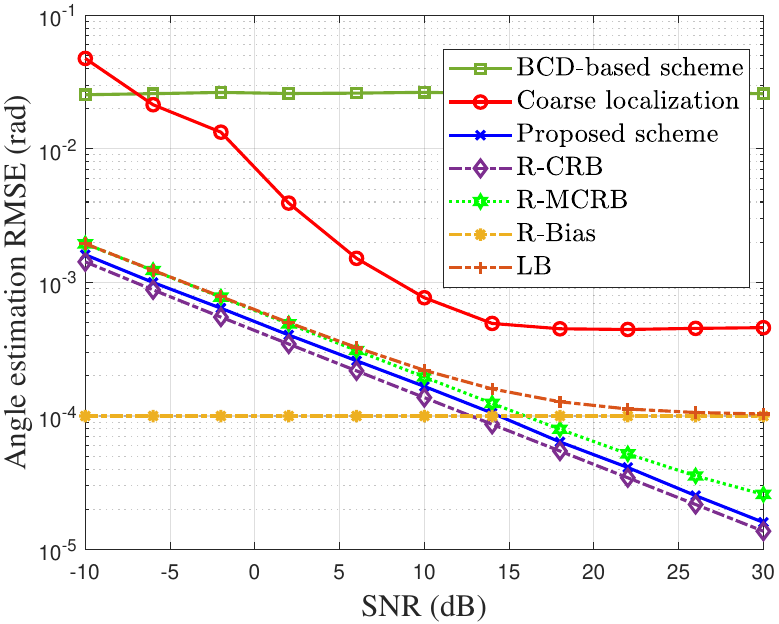}} \hspace{40pt}
		\subfigure[Range estimation RMSE versus SNR.]{
			\includegraphics[width=0.35\linewidth]{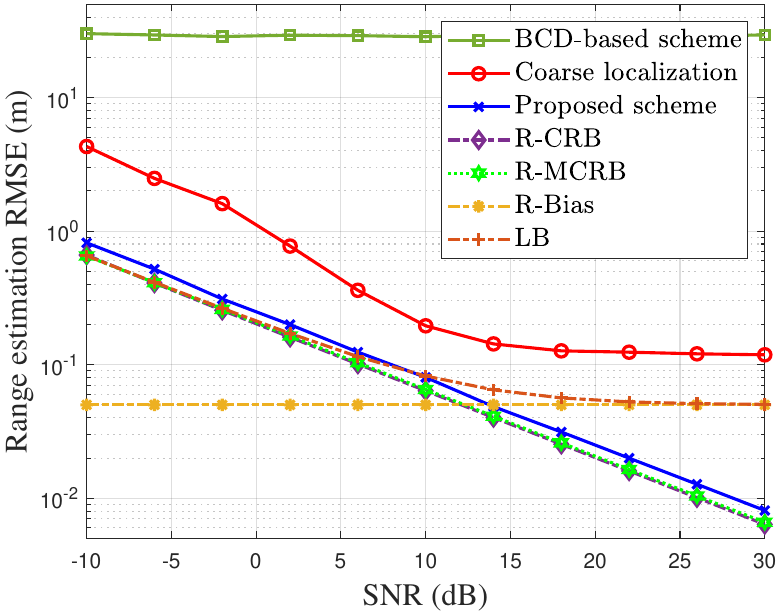}
			\label{fig:rangermsevssnr}} 
		\caption{\centering RMSE of multi-target localization versus SNR.}
		\label{fig:RMSE_vsSNR}\vspace{-10pt}
	\end{figure*}

	\begin{figure*}[t]	
		\centering
		\subfigure[Angle estimation RMSE versus number of antennas.]{ \label{fig:figanglen}	\includegraphics[width=0.35\linewidth]{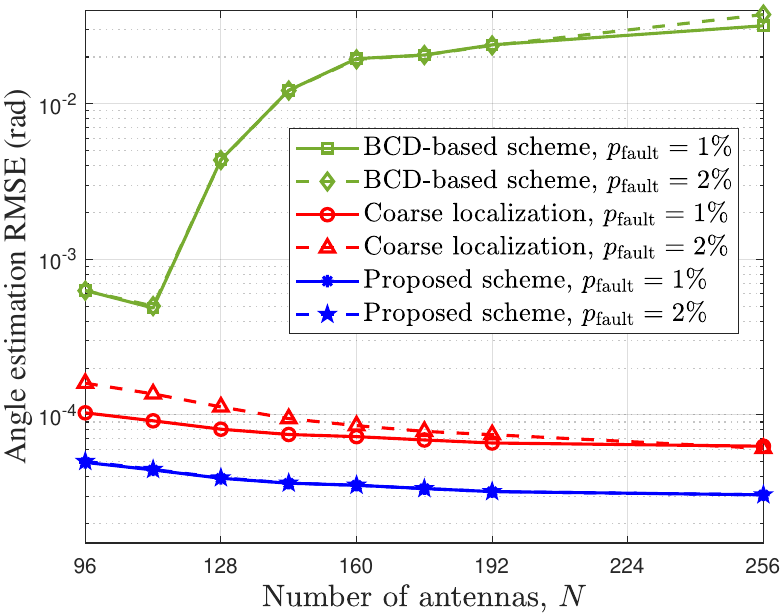}} \hspace{40pt}
		\subfigure[Range estimation RMSE versus number of antennas.]{
			\includegraphics[width=0.35\linewidth]{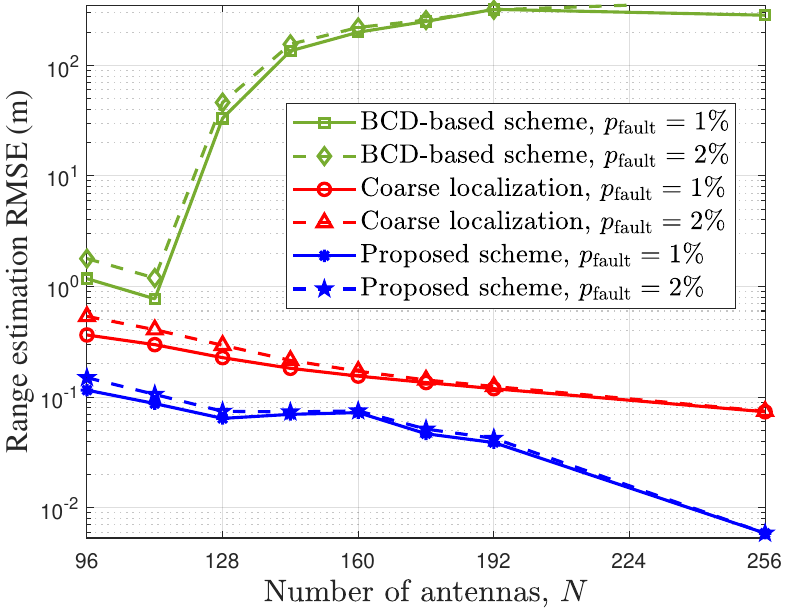} \label{fig:figrangen}
		}
		\caption{\centering RMSE of multi-target localization versus number of antennas.}	\label{fig:RMSE_vsN}\vspace{-10pt}
	\end{figure*}
	

	\vspace{-7pt}
	\subsection{Angle/Range Estimation RMSE} 
	First, we plot in Fig.~\ref*{fig:anglermsevssnr} the curves of angle estimation RMSE of different target localization schemes versus the received SNR varied from $-10$ dB to $30$ dB. 
	It is observed that as the SNR increases, the angle estimation RMSE of the proposed scheme monotonically decreases and approaches R-CRB, verifying the effectiveness of the proposed localization scheme.
	Next, for the coarse localization scheme, the angle estimation RMSE firstly decreases with SNR and then saturates when SNR is larger than a threshold (e.g., 15 dB).
	{ The saturation of RMSE in the coarse localization scheme is attributed to two main factors: the estimation accuracy of subarrays and the array configuration for triangulation.}
	In addition, the BCD-based scheme suffers from a high RMSE over different SNRs, which is expected since this scheme tends to converge to a wrong low-quality solution when targets are close to the XL-array (see Section~\ref{Sec3}).
	In contrast, our proposed scheme significantly outperforms the BCD-based scheme.
	Moreover, one can observe that the Root-LB tends to approach Root-bias quantified by the first term $\mathbf{Bias}$ in~\eqref{LB}, indicating that LB is dominated by $\mathbf{MCRB}$ in the low-SNR regime and by $\mathbf{Bias}$ in the high-SNR regime.
	The $\mathbf{Bias}$ results from the misspecification between the true model with HIs and the assumed model without HIs, i.e., localization performance degradation caused by ignoring HIs.
	
	Fig.~\ref*{fig:rangermsevssnr} shows the RMSE of range estimation versus received SNR by different schemes. 
	One can observe that the gap between the BCD-based scheme and the proposed scheme in the range estimation RMSE is larger than that in the angle estimation RMSE.
	This is because the BCD-based scheme incurs a significant deviation in range estimation even when the angle estimation error is slight (see Fig.~\ref{fig:Cost}).
	In contrast, our proposed scheme achieves a very low range estimation RMSE, since it can effectively tackle the issue of the BCD-based scheme.
	Other observations are similar to those in Fig.~\ref*{fig:anglermsevssnr}.


	In Fig.~\ref*{fig:figanglen}, we present the angle estimation RMSE versus the number of antennas $N$ under two cases with different fault probabilities (i.e., $p_{\text{fault}} = 1\%$ and $p_{\text{fault}} = 2\%$).
	First, given the fault probability $p_{\text{fault}}$, the angle estimation RMSE of both the coarse localization and proposed scheme consistently decreases as the number of antennas increases, while the angle estimation RMSE of the proposed scheme is smaller than those of other benchmark schemes.
	This is because the increased number of antennas results in a larger array aperture, which sharpens the peaks of MUSIC spectrum and improves spatial resolution.
	{ On the other hand, the BCD-based scheme performs poorly and shows a counter-intuitive trend, whose angle estimation RMSE first decreases and then increases sharply as the number of antennas increases. 
	This is because, when the number of antennas $N$ is small, the spatial resolution improvement provided by a larger array aperture helps reduce the RMSE.
	However, the near-field effect becomes more~prominent as the number of antennas becomes larger, causing the BCD-based scheme to converge to incorrect estimation values, and hence resulting in severe performance degradation in near-field localization. }
	In addition, when the fault probability $p_{\text{fault}}$ increases from $1\%$ to $2\%$, the proposed scheme still achieves superior performance with only a marginal rise in angle estimation RMSE, highlighting its robustness against HIs.
	The range estimation results shown in Fig.~\ref*{fig:figrangen} are generally consistent with the above findings.
	Moreover, the range estimation RMSE of the proposed scheme steadily decreases with the number of antennas, confirming its estimation effectiveness in the range domain. 
	
	\begin{figure}
		\centering
		\includegraphics[width=0.7\linewidth]{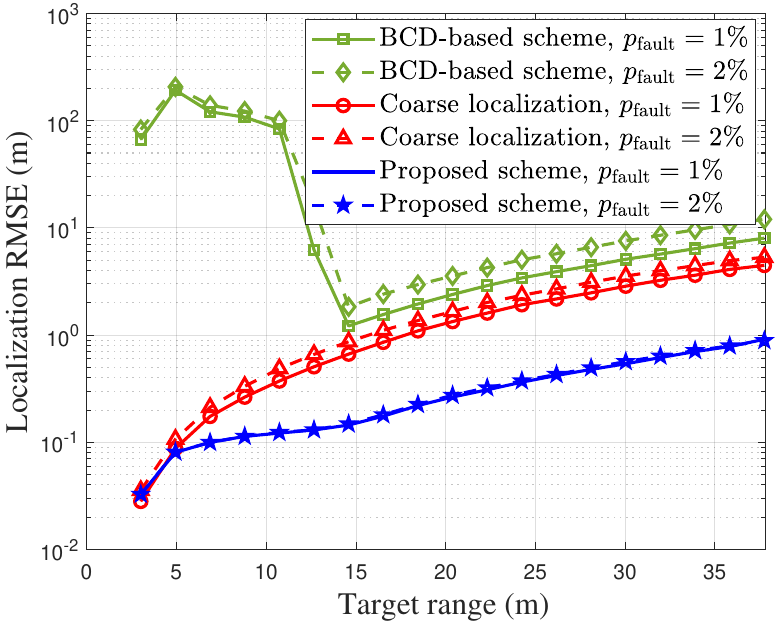}\vspace{-2pt}
		\caption{\centering Localization RMSE versus range of single target.}
		\label{fig:localizationvsdistance0703} \vspace{-10pt}
	\end{figure}
	\begin{figure}
		\centering
		\includegraphics[width=0.7\linewidth]{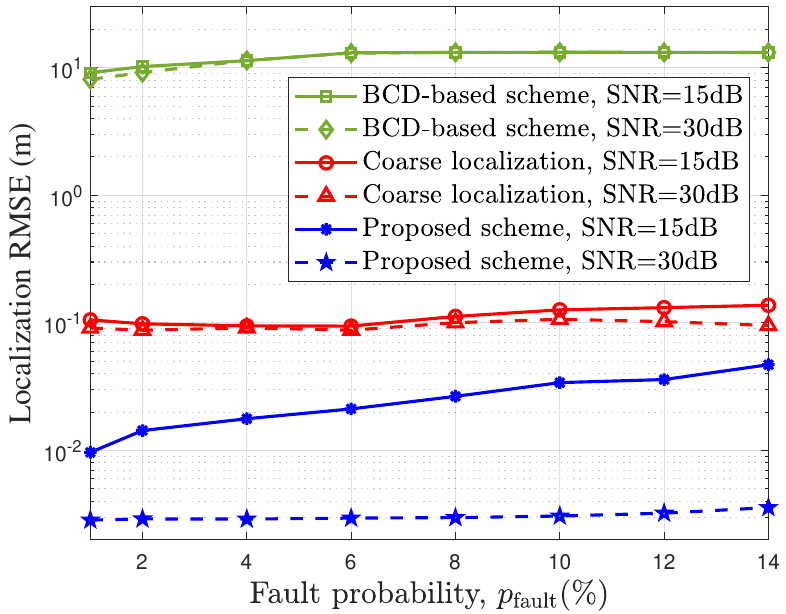} \vspace{2pt}
		\caption{\centering Localization RMSE versus fault probability.}
		\label{fig:figlocalizationvsp}\vspace{-10pt}
	\end{figure}

	\vspace{-12pt}
	\subsection{Effect of System Parameters} \vspace{-2pt}
	Furthermore, we plot in Fig.~\ref{fig:localizationvsdistance0703} the curves of the localization RMSE versus the range of a single target, where its angle is fixed as $\frac{\pi}{12} $ rad.
	It is observed that the localization RMSE of both the coarse localization and proposed schemes gradually increases as the target moves away from the XL-array.
	This is expected since the information available only for range estimation in near-field localization is the spherical wavefront curvature embedded in the near-field steering vector, which becomes more pronounced as the target is closer to the XL-array.
	Moreover, in this case, the BCD-based scheme is unable to localize the target with high accuracy.

	In addition, to demonstrate the robustness of the proposed scheme, we show in Fig.~\ref{fig:figlocalizationvsp} the RMSE of localization against fault probability $p_{\rm{fault}}$ within two distinct SNR levels: 15 dB and 30 dB. 
	One can observe that the proposed scheme significantly outperforms both the BCD-based scheme and the coarse localization across different fault probabilities.
	Note that when $\mathsf{SNR}=15$ dB, the RMSE of our proposed scheme deteriorates slightly when the fault probability $p_{\mathrm{fault}}$ increases, while the RMSE is much smaller at $\mathsf{SNR}=30$ dB  even when the fault probability reaches 14\%.
	{ This means that even with high fault probability (e.g., 14\%), there are still enough intact contiguous subarrays to perform effective coarse localization, which validates the robustness of our approach under the sparsity assumption.
	} 

	Finally, we verify the performance gain of phase calibration by evaluating the normalized MSE (NMSE) of the estimated HI coefficient vector, which is defined as 
	$
	\mathrm{NMSE} = \frac{\Vert\check{\mathbf{c}} - \mathbf{c}\Vert_2^2}{\Vert\mathbf{c}\Vert_2^2}
	$,
	where $\check{\mathbf{c}}$ denotes the estimated HI coefficient vector.
	In Fig.~\ref{fig:fighivssnr}, we show the HI coefficient vector NMSE versus the SNR with different fault probabilities, i.e., $p_{\text{fault}} = 1\%$, $p_{\text{fault}} = 2\%$, and $p_{\text{fault}} = 4\%$.
	It is observed that the NMSE of the proposed scheme decreases with the SNR, which demonstrates that our proposed scheme can effectively leverage higher signal quality to achieve more accurate estimation of the HI coefficient vector.
	As expected, a larger fault probability leads to a larger NMSE. For instance, when $\mathsf{SNR} = 20$ dB, the NMSE of the proposed scheme increases when $p_{\text{fault}}$ is increased from 1\% to 2\% and then to 4\%. 
	Nevertheless, even at a high fault probability ($p_{\text{fault}} = 4\%$), the proposed scheme still shows significant performance improvement with increasing SNR.
	In contrast, the NMSE of the BCD-based scheme maintains at a high value across the entire SNR regime from -10 dB to 30 dB, since the BCD-based scheme is ineffective for HI vector estimation in the near-field region.
	%
	%
	%
	%

	\vspace{-0pt}
	\section{Conclusions}\label{Sec:Con}\vspace{-3pt}
	In this paper, we proposed an efficient HI-aware near-field localization method to detect faulty antennas and estimate the positions of targets for XL-array systems.
	To this end, we first detected faulty antennas, followed by correcting the phase of the detected antennas.
	Subsequently, we devised an efficient fine-grained localization method to accurately estimate the angles and ranges of targets based on a calibrated fully XL-array.
	Finally, numerical results validated that our proposed method achieves more accurate target localization as compared to various benchmark schemes.
	
	\begin{figure}
	\centering
	\includegraphics[width=0.7\linewidth]{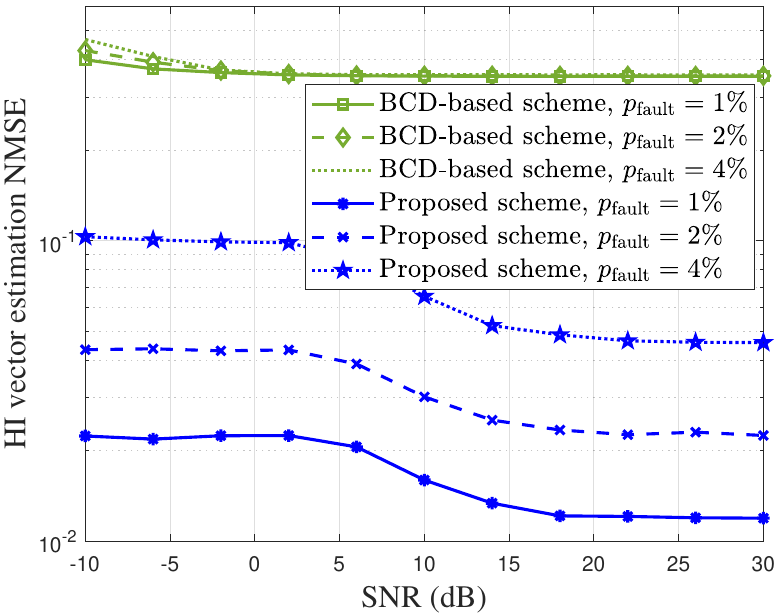} \vspace{-0pt}
	\caption{\centering HI coefficient vector estimation NMSE versus SNR.}
	\label{fig:fighivssnr} \vspace{-10pt}
\end{figure}

	\begin{appendices}
		\section{Proof of Lemma~\ref{Lemma1}}\label{App1} 
		We denote $\mathbf{E}_{k,q} = \mathbf{e}_{k,q}\mathbf{e}_{k,q}^T$, which is a symmetric and idempotent matrix, i.e., $\mathbf{E}_{k,q}^T = \mathbf{E}_{k,q}$ and $\mathbf{E}_{k,q}^T\mathbf{E}_{k,q} =  \mathbf{E}_{k,q}$ with $\mathbf{e}_{k,q}^T\mathbf{e}_{k,q} = \left|\mathbf{e}_{k,q} \right|^2 = 1$.
		As such, $\mathbf{Q}_{k,q} = \mathbf{I} - \mathbf{E}_{k,q}$ is also an idempotent matrix, i.e.,
		\begin{align}
			\mathbf{Q}_{k,q}^T\mathbf{Q}_{k,q} &= \left( \mathbf{I} - \mathbf{E}_{k,q} \right)^T\left( \mathbf{I} - \mathbf{E}_{k,q} \right) \nonumber\\[-2pt]
			& = \mathbf{I} - \mathbf{E}_{k,q}- \mathbf{E}_{k,q} +\mathbf{E}_{k,q}^T\mathbf{E}_{k,q} \nonumber\\[-2pt]
			& = \mathbf{I} - \mathbf{E}_{k,q} = \mathbf{Q}_{k,q}.
		\end{align} 
		Based on the above, the objection function of Problem \textbf{(P7)} can be re-expressed as 
		\begin{align} \label{Cost_coarse_localization}
			&\sum_{q=1}^{Q}\left|\boldsymbol{\xi}_k-\left(\mathbf{p}_q + \mathbf{e}_{k,q}( \boldsymbol{\xi}_k-\mathbf{p}_q )^T \mathbf{e}_{k,q}\right)\right|^2 \nonumber\\
			&=  \sum_{q=1}^{Q} \left(\boldsymbol{\xi}_k - \mathbf{p}_q\right)^T \mathbf{Q}_{k,q}\left(\boldsymbol{\xi}_k - \mathbf{p}_q\right).
		\end{align}
		Therefore, the optimal solution to Problem \textbf{(P7)} can be obtained as in~\eqref{P7_solution} based on its first-order optimality condition, hence completing the proof of Lemma~\ref{Lemma1}.
		
		\vspace{-4pt}
		\section{Derivatives in~\eqref{eq_mcrb_def} and~\eqref{FIM_CRB} }\label{App2} \vspace{-2pt}
		The derivatives in~\eqref{eq_mcrb_def} and~\eqref{FIM_CRB} can be re-expressed by using the derivative of $\mathbf{a}(\theta_i,r_i)$.
		In~\eqref{eq_mcrb_def}, 
		$\frac{\partial\tilde{\boldsymbol{\mu}}_t(\boldsymbol{\eta})}{\partial\eta_i} =  \mathbf{a}(\theta_i,r_i)s_{i,t}, 1\le i\le K$,
		$\frac{\partial\tilde{\boldsymbol{\mu}}_t(\boldsymbol{\eta})}{\partial\eta_i} = \jmath  \mathbf{a}(\theta_i,r_i) s_{i,t}, K+1\le i\le 2K$,
		$\frac{\partial\tilde{\boldsymbol{\mu}}_t(\boldsymbol{\eta})}{\partial\eta_i} = \frac{\partial \mathbf{a}(\theta_i,r_i)}{\partial \theta_i}\beta_i s_{i,t}, 2K+1\le i\le 3K$ and  $\frac{\partial\tilde{\boldsymbol{\mu}}_t(\boldsymbol{\eta})}{\partial\eta_i} = \frac{\partial \mathbf{a}(\theta_i,r_i)}{\partial r_i}\beta_i s_{i,t}, 3K+1 \le i \le 4K$.
		In~\eqref{FIM_CRB}, each derivative has similar form with~\eqref{eq_mcrb_def}, thus is omitted here.
		The elements of derivative for $\mathbf{a}(\theta_i,r_i)$ are give by  
		\begin{align}
			&\Big[ \frac{\partial \mathbf{a}(\theta_i,r_i)}{\partial \theta_i} \Big]_n \!\!\!\!= \!\! \frac{\jmath2\pi}{\sqrt{N}\lambda}\!\! \Big(\! l_n\! \cos(\theta_i) \!+\! \frac{l_n^2  \cos(\theta_i)\sin(\theta_i)}{r_i} \!\Big)\!a(\theta_i,r_i)\!,\nonumber\\ \nonumber
			&\Big[ \frac{\partial \mathbf{a}(\theta_i,r_i)}{\partial r_i} \Big]_n=  \frac{\jmath2\pi}{\sqrt{N}\lambda} \Big(  \frac{l_n^2  \cos^2(\theta_i)}{2r^2_i} \Big)a(\theta_i,r_i),\\
			&\Big[ \frac{\partial^2 \mathbf{a}(\theta_i,r_i)}{\partial \theta^2_i} \Big]_n\!\!\!=\!\!  \Big[-\frac{4\pi^2}{{N}\lambda^2} \Big( l_n \cos(\theta_i) + \frac{l_n^2  \cos(\theta_i)\sin(\theta_i)}{r_i} \Big)^2 \Big.\nonumber\\ 
			& +\!\! \Big.\frac{\jmath2\pi}{\sqrt{N}\lambda} \!\Big(\! -l_n \sin(\theta_i) \!-\! \frac{l_n^2  \sin^2(\theta_i)}{r_i} \!+\! \frac{l_n^2  \cos^2(\theta_i)}{r_i} \Big)  \!\Big] \!a(\theta_i,r_i),\nonumber\\
			&\Big[ \frac{\partial^2 \mathbf{a}(\theta_i,r_i)}{\partial \theta_i \partial r_i} \Big]_n \!\!=\!\! \Big[ \frac{\partial^2 \mathbf{a}(\theta_i,r_i)}{\partial r_i \partial \theta_i } \Big]_n \!\!=\!\!  \Big[- \frac{\jmath2\pi l_n^2 \cos(\theta_i) \sin(\theta_i)}{\sqrt{N}\lambda r^2_i}     \Big.\nonumber\\ &  \Big. -\! \Big(\! l_n \cos(\theta_i) \!+\! \frac{l_n^2  \cos(\theta_i)\sin(\theta_i)}{r_i}\! \Big)\!\frac{2\pi^2 l_n^2  \cos^2(\theta_i)}{{N}\lambda^2 r^2_i}  \Big]\! a(\theta_i,r_i),\nonumber\\ 
			&\Big[ \frac{\partial^2 \mathbf{a}(\theta_i,r_i)}{ \partial r^2_i} \Big]_n\!\!\!=\!\!  \Big[\!-  \frac{\jmath2\pi l_n^2 \cos(\theta_i) }{\sqrt{N}\lambda r^3_i}     \!-\!  \frac{\pi^2l_n^4  \cos^4(\theta_i)}{{N}\lambda^2 r^4_i}\!  \Big]\! a(\theta_i,r_i).\nonumber
		\end{align}
	\end{appendices}

	\bibliographystyle{IEEEtran}
	\bibliography{Ref_title.bib}
	
\end{document}